\documentclass{llncs}
\usepackage{amsmath,amsfonts,amssymb}
\usepackage{multirow}
\usepackage{multicol}
\usepackage{mathtools}
\usepackage{url}
\title{
    The \(\QQ\)-curve Construction
    for Endomorphism-Accelerated Elliptic Curves
}
\author{Benjamin Smith}
\institute{
    INRIA
    \emph{and}
    \'Ecole polytechnique
    \\
    \'Equipe-projet GRACE, INRIA Saclay--\^Ile-de-France
    \\
    Laboratoire d'Informatique 
    de l'\'Ecole polytechnique (LIX) 
    \\
    B\^atiment Alan Turing, 1 rue Honor\'e d'Estienne d'Orves,
    91120 Palaiseau, France
}

\newcommand{\ZZ}{\mathbb{Z}}
\newcommand{\QQ}{\mathbb{Q}}
\newcommand{\QQbar}{\overline{\mathbb{Q}}}
\newcommand{\Kbar}{\overline{K}}
\newcommand{\FF}{\mathbb{F}}
\newcommand{\FFbar}{\overline{\mathbb{F}}}
\newcommand{\PP}{\mathbb{P}}

\newcommand{\EC}{\mathcal{E}}
\newcommand{\ECK}{\widetilde{\mathcal{E}}}
\newcommand{\phiK}{\widetilde{\phi}}
\newcommand{\G}{\mathcal{G}}

\newcommand{\Lattice}{\mathcal{L}}
\newcommand{\vv}[1]{\mathbf{#1}}
\newcommand{\End}{\mathrm{End}}
\newcommand{\Gal}{\mathrm{Gal}}

\newcommand{\trace}[1]{{t_{#1}}}

\newcommand{\subgrp}[1]{\langle{#1}\rangle}
\newcommand{\roundoff}[1]{\lfloor{#1}\rceil}

\newcommand{\conj}[2][\sigma]{{{}^{#1}{#2}}}
\newcommand{\dualof}[1]{{{#1}^{\dagger}}}
\newcommand{\twist}[2]{#2^{#1}}
\newcommand{\twistiso}[1]{\delta({#1})}

\newcommand{\Legendre}[2]{\left({#1}\big/{#2}\right)}

\newcommand{\Oh}{O}

\newcommand{\AM}{{A_{2,\Delta}^\mathrm{M}(s)}}
\newcommand{\BM}{{B_{2,\Delta}^\mathrm{M}(s)}}

\renewcommand{\tabcolsep}{2pt}

\newtheorem{algorithm}{Algorithm}

\begin{document}
\maketitle

\begin{abstract}
    We give a detailed account of the use of \(\QQ\)-curve reductions 
    to construct elliptic curves over \(\FF_{p^2}\)
    with efficiently computable endomorphisms,
    which can be used
    to accelerate elliptic curve-based cryptosystems
    in the same way as Gallant--Lambert--Vanstone (GLV) 
    and Galbraith--Lin--Scott (GLS) endomorphisms.
    Like GLS (which is a degenerate case of our construction), 
    we offer the advantage over GLV of selecting from 
    a much wider range of curves, and thus finding secure group orders 
    when \(p\) is fixed for efficient implementation.
    Unlike GLS, we also offer the possibility of constructing
    twist-secure curves.
    We construct 
    several one-parameter families of elliptic curves over
    \(\FF_{p^2}\) equipped with efficient endomorphisms
    for every \(p > 3\),
    and exhibit examples of twist-secure curves 
    over \(\FF_{p^2}\)
    for the efficient Mersenne prime \(p = 2^{127}-1\).
\end{abstract}

\begin{keywords}
    Elliptic curve cryptography; endomorphism; exponentiation; 
    GLS; GLV; \(\QQ\)-curves; scalar decomposition; scalar multiplication
\end{keywords}

\bigskip

\emph{%
This is an extended treatment 
of the curves and techniques
introduced in~\cite{Smith-asiacrypt},
including two more families of curves,
more efficient scalar decompositions,
and more detail on exceptional CM curves and 
4-dimensional decompositions.
}

\section{
    Introduction
}
\label{sec:intro}

Let \(\EC\) be an elliptic curve over a finite field \(\FF_{q}\),
and let \(\G \subseteq \EC(\FF_{q})\)
be a cyclic subgroup of prime order \(N\).
When implementing cryptographic protocols in~\(\G\),
the fundamental operation is
\emph{scalar multiplication}
(or \emph{exponentiation}):
\begin{center}
    Given \(P\) in \(\G\) and \(m\) in \(\ZZ\),
    compute 
    \([m]P := \underbrace{P\oplus\cdots\oplus P}_{m \text{ times}} \).
\end{center}

The literature on general scalar multiplication algorithms is vast, 
and we will not explore it in detail here 
(see~\cite[\S2.8,\S11.2]{Galbraith} 
and~\cite[Chapter 9]{Handbook}
for introductions to exponentiation and 
multiexponentiation algorithms).
For our purposes,
it suffices to note that
the dominant factor in scalar multiplication time
using conventional algorithms
is the bitlength \(\lceil\log_2|m|\rceil\) of \(m\).
As a basic example,
we may compute \([m]P\)
using a variant of the classic binary method,
which requires at most \(\lceil\log_2|m|\rceil\) doublings
and (in the worst case) about as many addings
in~\(\G\)
(in general, \(\log_2|m| \sim \log_2N \sim \log_2q\)).

Suppose \(\EC\) is 
equipped with an efficient
\(\FF_{q}\)-endomorphism~\(\psi\). 
By \emph{efficient},
we mean that we can compute the image \(\psi(P)\)
of any point \(P\) in \(\EC(\FF_{q})\) 
for the cost of \(\Oh(1)\) operations in~\(\FF_{q}\).
In practice, we want this to cost no more than a few doublings in
\(\EC(\FF_{q})\).

Assume \(\psi(\G) \subseteq \G\),
or equivalently,
that \(\psi\) restricts to an endomorphism of \(\G\).\footnote{
    The assumption is satisfied, almost by default,
    in the context of classical discrete log-based
    cryptosystems.
    If \(\psi(\G) \not\subseteq \G\),
    then \(\EC[N](\FF_{q}) = \G+\psi(\G) \cong (\ZZ/N\ZZ)^2\),
    so \(N^2\mid\#\EC(\FF_{q})\)
    and \(N\mid q-1\);
    such \(\EC\) are cryptographically inefficient,
    and discrete logs in \(\G\) are vulnerable to the
    Menezes--Okamoto--Vanstone 
    and Frey--R\"uck
    reductions~\cite{Menezes--Okamoto--Vanstone,Frey--Ruck}.
    However, 
    the assumption
    should be verified carefully
    in the context of pairing-based cryptography,
    where \(\G\) and \(\psi\) with \(\psi(\G) \not\subseteq \G\)
    arise naturally.
}
Now \(\G\) is a finite cyclic group, isomorphic to \(\ZZ/N\ZZ\),
and every endomorphism of \(\ZZ/N\ZZ\) is just an integer
multiplication modulo~\(N\).
Hence, \(\psi\) acts on \(\G\) as multiplication by some integer eigenvalue
\(\lambda_\psi\): that is,
\[
    \psi|_\G = [\lambda_\psi]_\G 
    \qquad
    \text{for some }
    {-N}/2 < \lambda_\psi \le N/2
    \ .
\]
The eigenvalue \(\lambda_\psi\) is a root 
of the characteristic polynomial of \(\psi\)
in \(\ZZ/N\ZZ\).

Returning to the problem of scalar multiplication:
we want to compute \([m]P\).
Rewriting \(m\) as
\[
    m = a + b\lambda_\psi \pmod N
\]
for some \(a\) and \(b\),
we can compute \([m]P\) using the relation
\[
    [m]P = [a]P \oplus [b\lambda_\psi]P = [a]P \oplus [b]\psi(P) 
\]
and a 2-dimensional multiexponentation
such as Straus's algorithm~\cite{Straus},
which has a loop length of \(\log_2\|(a,b)\|_\infty\):
that is, \(\log_2\|(a,b)\|_\infty\) doubles 
and at most as many adds\footnote{
    Straus's algorithm
    serves as a simple and convenient reference example here:
    like all multiexponentiation
    algorithms, 
    its running time depends essentially on \(\log_2\|(a,b)\|_\infty\).
    It is not the fastest multiexponentiation,
    nor is it uniform or constant-time;
    as such, it is not recommended for real-world implementations.
}
(recall \(\|(a,b)\|_\infty = \max(|a|,|b|)\)).
If \(|\lambda_\psi|\) is not too small,
then we can easily find \((a,b)\)
such that \(\log_2\|(a,b)\|_\infty\)
is roughly \(\frac{1}{2}\log_2N\)
(we remove the ``If'' and the
``roughly'' for our \(\psi\) in \S\ref{sec:decompositions}.)

The endomorphism therefore lets us 
replace conventional \(\log_2N\)-bit scalar multiplications 
with \((\frac{1}{2}\log_2N)\)-bit multiexponentiations.
In basic binary methods
this means halving the loop length,
cutting the number of doublings in half.

Of course, in practice we are not halving the execution time.
The precise speedup 
depends on a variety of factors,
including
the choice of exponentiation and multiexponentiation
algorithms, the cost of computing \(\psi\), 
and the cost of doublings and addings in terms of bit operations---to 
say nothing of the cryptographic protocol,
which may prohibit some other conventional speedups.
For example: 
in~\cite{GLS}, Galbraith, Lin, and Scott 
report experiments where cryptographic operations
on GLS curves required between 70\% and 83\% of the time required for 
the previous best practice curves---with the variation
depending on the architecture, 
the underlying curve arithmetic, and the protocol in question.

To put this technique into practice,
we need a source of cryptographic elliptic curves 
equipped with efficient endomorphisms.
In the large characteristic case\footnote{
    We are primarily interested in the large characteristic case,
    where \(q = p\) or \(p^2\),
    so we do not discuss \(\tau\)-adic or Frobenius expansion-style
    techniques (see~\cite{Smart}, \cite[\S5]{Lange}).
},
there are two archetypal constructions:
\begin{enumerate}
    \item The classic \emph{Gallant--Lambert--Vanstone} (GLV)
        construction~\cite{GLV}
        uses 
        elliptic curves with explicit complex multiplication (CM)
        by quadratic orders with tiny discriminants.
        These curves can be found by reducing CM curves over number fields
        modulo suitable primes.
    \item The more recent 
        \emph{Galbraith--Lin--Scott} (GLS) construction~\cite{GLS}.
        Here, curves over \(\FF_p\) are viewed over \(\FF_{p^2}\);
        the \(p\)-power sub-Frobenius induces an extremely efficient
        endomorphism on the quadratic twist
        (which can have prime order).
\end{enumerate}
GLV and GLS 
have been 
combined to give higher-dimensional
variants 
for elliptic curves
(\cite{Longa--Sica}, \cite{ZHXS}),
and extended to hyperelliptic curves 
(\cite{BCHL}, \cite{Kohel--Smith}, \cite{SCQ}, \cite{Takashima}).

\paragraph{New endomorphisms from \(\QQ\)-curves.}

This work develops 
a new source of elliptic curves over \(\FF_{p^2}\)
with efficient endomorphisms: reductions of quadratic \(\QQ\)-curves.
\begin{definition}
    \label{def:QQ-curves}
    A \emph{quadratic \(\QQ\)-curve of degree \(d\)}
    is an elliptic curve \(\EC\) 
    \emph{without} CM,
    defined over a quadratic number field
    \(K\),
    such that there exists an isogeny of degree \(d\)
    from \(\EC\) to its Galois conjugate \(\conj{\EC}\)
    (formed by applying \(\sigma\) to the coefficients of the defining
    equation of \(\EC\)),
    where \(\subgrp{\sigma} = \Gal(K/\QQ)\).
\end{definition}

\(\QQ\)-curves are well-established objects of interest in number
theory, where they have formed a natural setting for 
generalizations of the Modularity Theorem.
We recommend Ellenberg~\cite{Ellenberg}
for an excellent introduction to the theory.
 
Our application of quadratic \(\QQ\)-curves is somewhat more prosaic:
given a \(d\)-isogeny \(\ECK \to \conj{\ECK}\) over a quadratic field,
we reduce modulo an inert prime \(p\)
to obtain an isogeny \(\EC\to\conj{\EC}\) over \(\FF_{p^2}\).
We then exploit the fact that the \(p\)-power Frobenius isogeny
maps \(\conj{\EC}\) back onto \(\EC\);
composing with the reduced \(d\)-isogeny,
we obtain an endomorphism of \(\EC\)
of degree \(dp\).
For efficiency,
\(d\) must be small;
happily, for small values of \(d\),
Hasegawa has written down universal one-parameter families of \(\QQ\)-curves~\cite{Hasegawa}.
We thus obtain one-parameter families of elliptic curves 
over~\(\FF_{p^2}\) 
equipped with efficient non-integer endomorphisms.\footnote{
    While our curves are defined over an extension field,
    the extension degree is only 2, so Weil descent
    attacks offer no advantage when solving DLP instances
    (cf.~\cite[\S9]{GLS}).
}

For concrete examples,
we consider \(d = 2, 3, 5\), and \(7\)
in~\S\ref{sec:degree-2}, \S\ref{sec:degree-3}, 
\S\ref{sec:degree-5}, and~\S\ref{sec:degree-7}, respectively.
We define our curves in short Weierstrass form
for maximum generality and flexibility,
but we also give isomorphisms to
Montgomery, twisted Edwards, and Doche--Icart--Kohel models 
where appropriate.

\paragraph{Comparison with GLV.}
Like GLV, our method involves reducing curves defined over number fields
to obtain curves over finite fields with explicit CM.
However, we emphasise a profound difference:
in our method, the curves over number fields 
generally do not have CM themselves.

GLV curves are necessarily isolated examples---and the really
useful examples are extremely limited in number
(cf.~\cite[App.~A]{Longa--Sica}).
The scarcity of GLV curves
is their Achilles' heel:
as noted in~\cite{GLS},
if~\(p\) is fixed
then there is no guarantee that there will exist a GLV curve 
with prime (or almost-prime) order over~\(\FF_{p}\).
While we discuss this phenomenon further in \S\ref{sec:CM},
it is instructive to consider the example discussed
in~\cite[\S 1]{GLS}:
the most efficient GLV curves have CM discriminants \(-3\) and \(-4\).
If we are working at the 128-bit security level,
then taking \(p = 2^{255}-19\)
allows particularly fast arithmetic in \(\FF_p\).
But the largest prime factor of the order of a curve
over \(\FF_p\)
with CM discriminant \(-4\) (resp.~\(-3\))
has \(239\) (resp.~\(230\)) bits:
using these curves wastes 9 (resp.~13) potential bits of security.
In fact, we are lucky with \(-3\) and \(-4\):
for all of the other discriminants offering endomorphisms of degree at most 3,
we can do no better than a 95-bit prime factor,
which represents a catastrophic 80-bit loss of relative security.

In contrast, our construction yields true families of curves,
covering \({\sim p}\) isomorphism classes over \(\FF_{p^2}\)
for any choice of \(p\).
This gives us
a vastly higher probability of finding
secure curves over practically important fields.

\paragraph{Comparison with GLS.}

Like GLS, we construct curves over \(\FF_{p^2}\)
equipped with an inseparable endomorphism.
And like GLS, 
each of our families offer around~\(p\) distinct isomorphism
classes of curves, making it easy to find secure group orders
when \(p\) is fixed.

But unlike GLS, our curves have \(j\)-invariants in \(\FF_{p^2}\)
and not \(\FF_p\):
they are not isomorphic to, or twists of, subfield curves.
This allows us to find twist-secure curves:
that is, curves with secure orders whose twists also have secure orders.

Twist-security is crucial
in many elliptic curve-based protocols
(including \cite{Kaliski-1}, \cite{Kaliski-2}, 
\cite{Boyd--Montague--Nguyen}, \cite{Moeller},
and~\cite{Chevassut--Fouque--Gaudry--Pointcheval});
it is particularly important in modern \(x\)-coordinate-only 
ECDH implementations,
such as Bernstein's \texttt{Curve25519} software~\cite{Bernstein}
(which anticipates the kind of fault attack detailed 
in~\cite{Fouque--Lercier--Real--Valette}).
GLS curves cannot be used in any of these constructions.

As we will see in \S\ref{sec:construction},
our construction degenerates to GLS when \(d = 1\).
Our construction is therefore a sort of 
generalized GLS---though it is not the higher-degree generalization 
anticipated by Galbraith, Lin, and Scott themselves,
which composes a \(p\)-power sub-Frobenius endomorphism with a non-rational separable
isogeny and its dual isogeny
(cf.~\cite[Theorem~1]{GLS}).

\paragraph{Acknowledgements.}
The author thanks Craig Costello, 
H\"useyin H{\i}\c{s}{\i}l,
Fran\c{c}ois Morain, and Charlotte Scribot
for their help and advice throughout this project.

\section{
    Notation and Elementary Constructions
}
\label{sec:notation}

Throughout, we work over fields of characteristic not 2 or 3.
Let 
\[
    \EC: y^2 = x^3 + Ax + B
\]
be an elliptic curve over such a field \(K\).
(Recall that for an equation in this form to define an elliptic
curve, the \emph{discriminant} \(-16(4A^3 + 27B^2)\) must be nonzero.)

\paragraph{Galois conjugates.}
For every automorphism \(\sigma\) of \(\Kbar\),
we have a conjugate curve 
\[
    \conj{\EC}: y^2 = x^3 + (\conj{A})x + (\conj{B})
    \ .
\]
If \(\phi: \EC_1\to\EC_2\) is an isogeny,
then 
we obtain a conjugate isogeny
\(\conj{\phi}:\conj{\EC_1}\to\conj{\EC_2}\)
by
applying \(\sigma\) 
to the defining equations of \(\phi\), \(\EC_1\), and \(\EC_2\).
We write \((p)\) for 
the \(p\)-th powering automorphism of \(\FFbar_p\).
We note that \((p)\) is trivial to compute
on \(\FF_{p^2} = \FF_{p}(\sqrt{\Delta})\),
since
\( \conj[(p)]{(a + b\sqrt{\Delta})} = a - b\sqrt{\Delta} \)
for all \(a\) and \(b\) in \(\FF_{p}\).

\paragraph{Quadratic twists and twisted endomorphisms.}
For every nonzero \(\lambda\) in \(\Kbar\),
the quadratic \emph{twist} of \(\EC\) by \(\lambda\) is
the curve over \(K(\lambda^2)\) defined by
\[
    \twist{\lambda}{\EC}: 
    y^2 = x^3 + \lambda^4Ax + \lambda^6B
    \ .
\]
The twisting isomorphism
\( \twistiso{\lambda}: \EC \to \twist{\lambda}{\EC} \)
is defined over \(K(\lambda)\) by
\[
    \twistiso{\lambda}
    : 
    (x,y) 
    \longmapsto
    (\lambda^2 x, \lambda^3 y) 
    \ .
\]
Observe that 
\(
    \twistiso{\lambda_1}\twistiso{\lambda_2}
    =
    \twistiso{\lambda_1\lambda_2}
\)
for any \(\lambda_1,\lambda_2\);
and 
\(\twistiso{-\lambda} = -\twistiso{\lambda}\).
For every \(K\)-endo\-morphism \(\psi\) of \(\EC\),
there is a twisted \(K(\lambda^2)\)-endomorphism 
\[
    \twist{\lambda}{\psi} 
    :=
    \twistiso{\lambda}\psi\twistiso{\lambda^{-1}}
    \ \text{ in }\ 
    \End(\twist{\lambda}{\EC})
    \ .
\] 
It is important to distinguish conjugates
(marked by left-superscripts)
from twists (marked by right-superscripts);
note that
\(\conj{(\twist{\lambda}{\EC})} =
\twist{\conj{\lambda}}{(\conj{\EC})}\)
for all \(\lambda\) and \(\sigma\).

\paragraph{Quadratic twists over finite fields.}
If \(\mu\) is a nonsquare in \(K = \FF_{q}\),
then \(\twist{\sqrt{\mu}}{\EC}\) is a quadratic twist of~\(\EC\).
But \(\twist{\sqrt{\mu_1}}{\EC}\) and 
\(\twist{\sqrt{\mu_2}}{\EC}\) are \(\FF_{q}\)-isomorphic
for all nonsquares \(\mu_1\) and \(\mu_2\) in \(\FF_{q}\)
(the isomorphism \(\twistiso{\sqrt{\mu_1/\mu_2}}\) is defined over
\(\FF_{q}\) because \(\mu_1/\mu_2\) must be a square);
so, up to \(\FF_{q}\)-isomorphism, 
it makes sense to speak of \emph{the} quadratic twist.
We let \(\EC'\) denote the quadratic twist
when the choice of nonsquare is not important.
Similarly,
if \(\psi\) is an \(\FF_{q}\)-endomorphism of \(\EC\),
then \(\psi'\)
denotes the corresponding twisted \(\FF_{q}\)-endomorphism of \(\EC'\).

\paragraph{Traces and cardinalities.}
If \(K = \FF_{q}\),
then \(\pi_{\EC}\) denotes 
the \(q\)-power Frobenius endomorphism of~\(\EC\).
The characteristic polynomial
of \(\pi_{\EC}\) has the form
\[
    \chi_{\EC}(T) = T^2 - \trace{\EC}T + q 
    ;
\]
the 
\emph{trace} \(\trace{\EC}\) 
satisfies
the Hasse bound
\(|\trace{\EC}| \le 2\sqrt{q}\).
Recall
\(\#\EC(\FF_{q}) = q + 1 - \trace{\EC}\)
and
\(\trace{\EC'} = -\trace{\EC}\),
so
\begin{equation}
    \label{eq:twist-cardinalities}
    \#\EC(\FF_{q}) + \#\EC'(\FF_{q}) = 2(q + 1) \ .
\end{equation}

\paragraph{Explicit isogenies.}
Let \(\mathcal{S} \subset \EC\) be a finite subgroup defined over \(K\);
V\'elu's formul\ae{} (\cite[\S2.4]{Kohel}, \cite{Velu})
compute the explicit (normalized) quotient isogeny 
\[
    \phi: \EC \longrightarrow 
    \EC/\mathcal{S} : y^2 = x^3 + A_\mathcal{S}x + B_\mathcal{S}
    \ ,
\]
mapping
\( (x,y) \) to \( \left(\phi_x(x),y\phi_x'(x)\right) \)
for some \(\phi_x\) in \(K(x)\).
We will need explicit formul\ae{} 
for the cases \(\#\mathcal{S} = 2, 3, 5\), and \(7\).
If \(\mathcal{S} = \{0,(\alpha,0)\}\) 
has order \(2\),
then 
\begin{equation}
    \label{eq:Velu-2}
    A_\mathcal{S} = -4A - 15\alpha^2
    \ ,
    \quad 
    B_\mathcal{S} = B - 7\alpha(3\alpha^2 + A)
    \ ,
    \quad
    \text{and}
    \quad
    \phi_x(x) = x + \tfrac{3\alpha^2 + A}{x - \alpha}
    \ .
\end{equation}
If \(\mathcal{S}\) has
odd order \(d = 2e+1\),
then it is defined by a kernel polynomial
\(F(x) = \sum_{i=0}^ef_ix^{e-i}\)
(so \(F(x(P)) = 0\) if and only if \(P\) is in
\(\mathcal{S}\setminus\{0\}\));
and then
\begin{align}
    \label{eq:Velu-odd-A}
    A_\mathcal{S} 
    & = (1 - 10e)A - 30(f_1/f_0)^2 + 60f_2/f_0
    \ ,
    \\
    \label{eq:Velu-odd-B}
    B_\mathcal{S} 
    & = (1 - 28e)B + 28f_1/f_0 
        + 70(f_1/f_0)^3 - 210f_1f_2/f_0^2 + 210f_3/f_0
    \ ,
    \shortintertext{and}
    \label{eq:Velu-odd-map}
    \phi_x(x)
    & =
    (2e+1)x + 2\tfrac{f_1}{f_0}
    -4(x^3 + Ax + B)\big(\tfrac{F'(x)}{F(x)}\big)'
    - 2(3x^2 + A)\tfrac{F'(x)}{F(x)}
    \ .
\end{align}

\paragraph{Legendre symbols.}
The Legendre symbol \(\Legendre{n}{p}\)
is defined to be \(1\) if \(n\) is a square mod \(p\),
\(-1\) if \(n\) is not a square mod \(p\),
and \(0\) if \(p\) divides \(n\).

\paragraph{Reduced lattice bases.}
We work exclusively with the infinity norm \(\|\cdot\|_\infty\)
in this article
(recall \(\|(a,b)\|_\infty := \max(|a|,|b|)\)).
An ordered basis \([\vv{e}_1,\vv{e}_2]\)
of a lattice in~\(\ZZ^2\)
is \emph{reduced}
if
\begin{equation}
    \label{eq:reduced}
    \|\vv{e}_1\|_\infty 
    \le 
    \|\vv{e}_2\|_\infty 
    \le 
    \|\vv{e}_1 - \vv{e}_2\|_\infty 
    \le 
    \|\vv{e}_1 + \vv{e}_2\|_\infty 
    \ ;
\end{equation}
a reduced basis has
minimal length with respect to \(\|\cdot\|_\infty\)
(see~\cite{Kaib}).

\section{
    Quadratic \(\QQ\)-curves and their Reductions
}
\label{sec:construction}

Suppose \(\ECK/\QQ(\sqrt{\Delta})\) is a quadratic \(\QQ\)-curve 
of prime degree \(d\) (as in~Definition~\ref{def:QQ-curves}),
where \(\Delta\) is a discriminant prime to \(d\),
and let
\(\phiK: \ECK \to \conj{\ECK}\)
be the corresponding \(d\)-isogeny
(where \(\sigma\) is the conjugation of \(\QQ(\sqrt{\Delta})\) over
\(\QQ\)).
In general, \(\phiK\) is only defined over a quadratic
extension \(\QQ(\sqrt{\Delta},\gamma)\)
of \(\QQ(\sqrt{\Delta})\)
(cf.~\cite[Prop.~3.1]{Gonzalez}), but 
we can always reduce to the case where 
\(\gamma = \sqrt{\pm d}\)
(see~\cite[remark p.~385]{Gonzalez}).
Indeed, the \(\QQ\)-curves of degree~\(d\) 
that we treat below
all have \(\gamma = \sqrt{-d}\);
so to simplify matters, from now on we will 
\begin{center}
    \sl Assume \(\phiK\) is defined over \(\QQ(\sqrt{\Delta},\sqrt{-d})\).
\end{center}

Let \(p\) be a prime inert in \(\QQ(\sqrt{\Delta})\) 
(equivalently, \(\Delta\) is not a square in \(\FF_p\)),
of good reduction for \(\ECK\)
and prime to \(d\).
If \(\mathcal{O}\) is the ring of integers of \(\QQ(\sqrt{\Delta})\),
then 
\[
    \FF_{p^2} = \mathcal{O}/(p) = \FF_p(\sqrt{\Delta}) 
    \ .
\]
Looking at the Galois groups of our fields, we have a series of
injections
\[
    \subgrp{(p)}
    =
    \Gal(\FF_p(\sqrt{\Delta})/\FF_p) 
    \hookrightarrow 
    \Gal(\QQ(\sqrt{\Delta})/\QQ) 
    \hookrightarrow
    \Gal(\QQ(\sqrt{\Delta},\sqrt{-d})/\QQ)
    \ .
\]
The image of \((p)\) in \(\Gal(\QQ(\sqrt{\Delta})/\QQ)\) is \(\sigma\),
because \(p\) is inert in \(\QQ(\sqrt{\Delta})\).
We extend \(\sigma\) to 
the automorphism of \(\QQ(\sqrt{\Delta},\sqrt{-d})\)
that is the image of \((p)\):
that is,
\begin{equation}
    \label{eq:FFp2-conj}
    \conj{\big(
        \alpha + \beta\sqrt{\Delta} + \gamma\sqrt{-d} + \delta\sqrt{-d\Delta}
    \big)} 
    = 
    \alpha - \beta\sqrt{\Delta} 
    + \Legendre{-d}{p}\big(
            \gamma\sqrt{-d} - \delta\sqrt{-d\Delta}
    \big)
\end{equation}
for all \(\alpha, \beta, \gamma\), and \(\delta \in \QQ\).

Now let \(\EC/\FF_{p^2}\)
be the reduction modulo \(p\) of \(\ECK\).
The curve \(\conj{\ECK}\) reduces to \(\conj[(p)]{\EC}\),
while the \(d\)-isogeny \(\phiK: \ECK\to\conj{\ECK}\)
reduces to a \(d\)-isogeny
\(\phi: \EC \to \conj[(p)]{\EC}\)
over~\(\FF_{p^2}\).

Applying \(\sigma\) to \(\phiK\),
we obtain a second \(d\)-isogeny
\(\conj{\phiK}:\conj{\ECK}\to\ECK\)
travelling in the opposite direction,
which reduces mod \(p\) to
a conjugate isogeny \(\conj[(p)]{\phi}: \conj[(p)]{\EC} \to \EC\)
defined over~\(\FF_{p^2}\).
Composing \(\conj{\phiK}\) with \(\phiK\)
yields endomorphisms \(\conj{\phiK}\circ\phiK\) of \(\ECK\)
and \(\phiK\circ\conj{\phiK}\) of \(\conj{\ECK}\),
each of degree~\(d^2\).
But (by definition) \(\ECK\) and \(\conj{\ECK}\)
do not have CM,
so all of their endomorphisms are integer multiplications;
and since the only integer multiplications of degree \(d^2\)
are \([d]\) and \([-d]\), 
we 
conclude that
\[
    \conj{\phiK}\circ\phiK = [\epsilon_p d]_{\ECK}
    \quad
    \text{and}
    \quad
    \phiK\circ\conj{\phiK} = [\epsilon_p d]_{\conj{\ECK}}
    \ ,
    \quad 
    \text{where}
    \quad 
    \epsilon_p \in \{\pm1\}
    \ .
\]
Technically, \(\conj{\phiK}\) and \(\conj[(p)]{\phi}\)
are---\emph{up to sign}---the dual isogenies of \(\phiK\) and \(\phi\),
respectively.
The sign \(\epsilon_p\) depends on \(p\): 
if \(\tau\) is the extension of \(\sigma\) to
\(\QQ(\sqrt{\Delta},\sqrt{-d})\) that is \emph{not} the image of \((p)\),
then \(\conj[\tau]{\phiK}\circ\phiK = [-\epsilon_p d]_{\ECK}\).
Reducing modulo~\(p\),
we see that
\[
    \conj[(p)]{\phi}\circ\phi 
    =
    [\epsilon_p d]_{\EC}
    \quad 
    \text{ and }
    \quad 
    \phi\circ\conj[(p)]{\phi} 
    =
    [\epsilon_p d]_{\conj[(p)]\EC}
    \ .
\]

The map
\((x,y)\mapsto(x^p,y^p)\)
defines \(p\)-isogenies
\[
    \pi_p: \conj[(p)]{\EC} \longrightarrow \EC 
    \quad \text{ and }\quad 
    \conj[(p)]{\pi_p}: \EC \longrightarrow \conj[(p)]{\EC}
    \ .
\]
Observe that
\(\conj[(p)]{\pi_p}\circ\pi_p = \pi_{\EC}\)
and \(\pi_p\circ\conj[(p)]{\pi_p} = \pi_{\conj[(p)]{\EC}}\).
Composing \(\pi_p\) with \(\phi\) yields a degree-\(pd\) endomorphism
\[
    \psi := \pi_p\circ\phi \in \End(\EC) 
    \ .
\]
We also obtain an \(\FF_{p^2}\)-endomorphism \(\psi'\)
on the quadratic twist \(\EC'\).

If \(d\) is very small, then 
\(\psi\) and \(\psi'\) are efficient 
because \(\phi\) is defined by polynomials of degree about \(d\),
and \(\pi_p\) acts as a simple conjugation on coordinates in~\(\FF_{p^2}\)
(as in Eq.~\eqref{eq:FFp2-conj}).
In this article we concentrate on prime \(d\le 7\).

\begin{theorem}
    \label{th:main}
    The endomorphisms \(\psi\) and \(\psi'\) satisfy
    \[
        \phantom{(-')\ .}
        \psi^2 = [\epsilon_p d]\pi_{\EC}
        \qquad
        \text{ and }
        \qquad
        (\psi')^2 = [-\epsilon_p d]\pi_{\EC'}
        \ ,
    \]
    respectively.
    There exists an integer \(r\) satisfying\footnote{
        We warn the reader that the integer \(r\) here corresponds to \(\epsilon_pr\)
        in~\cite{Smith-asiacrypt}.
    }
    \begin{equation}
        \label{eq:dr2-2p-et}
        dr^2 = 2p + \epsilon_p \trace{\EC} 
    \end{equation}
    such that
    \begin{equation}
        \label{eq:rpsi}
        [r]\psi = [p] + \epsilon_p\pi_{\EC}
        \qquad
        \text{ and }
        \qquad
        [r]\psi' = [p] - \epsilon_p\pi_{\EC'}
        \ ;
    \end{equation}
    the characteristic polynomial
    of both \(\psi\) and \(\psi'\) 
    is
    \[
        P_\psi(T)
        =
        P_{\psi'}(T)
        =
        T^2 - rdT + dp 
        \ .
    \]
\end{theorem}
\begin{proof}
    Clearly \( \pi_p\circ\phi = \conj[(p)]{\phi}\circ\conj[(p)]{\pi_p} \),
    so
    \[
        \psi^2 
        = 
        \pi_p\phi\pi_p\phi 
        =
        \pi_p\phi\conj[(p)]{\phi}\conj[(p)]{\pi_p}
        =
        \pi_p[\epsilon_p d]\conj[(p)]{\pi_p}
        =
        [\epsilon_p d]\pi_p\conj[(p)]{\pi_p}
        =
        [\epsilon_p d]\pi_{\EC}
        \ .
    \]
    Similarly,
    choosing a nonsquare \(\mu\) in \(\FF_{p^2}\), so \(\EC' =
    \twist{\sqrt{\mu}}{\EC}\) and \(\psi' = \twist{\sqrt{\mu}}{\psi}\), 
    we find
    \[
        (\psi')^2 
        = 
        \twistiso{\mu^{\frac{1}{2}}}\psi^2\twistiso{\mu^{-\frac{1}{2}}}
        =
        \twistiso{\mu^{\frac{1}{2}(1-p^2)}}[\epsilon_p d]\pi_{\EC'}
        =
        \twistiso{-1}[\epsilon_p d]\pi_{\EC'}
        =
        [-\epsilon_p d]\pi_{\EC'}
        \ .
    \]
    The degree of \(\psi\) (and \(\psi'\)) is \(dp\),
    so both have a characteristic polynomial in the form
    \(P_{\psi}(T) = T^2 - aT + dp\)
    for some integer \(a\).
    Hence,
    \begin{equation}
        \label{eq:a-eqn}
        [a]\psi = \psi^2 + [dp] = [\epsilon_pd]\pi + [dp]
        \ .
    \end{equation}
    Squaring both sides, replacing \(\psi^2\) with
    \([\epsilon_pd]\pi\),
    and then factoring out \([\epsilon_pd]\pi\),
    we find \( a^2 = 2dp + d\epsilon_p\trace{\EC} \).
    It follows that \(d\mid a^2\); but \(d\) is squarefree,
    so \(a = dr\) for some integer \(r\),
    whence Eq.~\eqref{eq:dr2-2p-et}
    and the characteristic polynomial.
    Putting \(a = dr\) in Eq.~\eqref{eq:a-eqn}
    yields \( [r]\psi = [p] + \epsilon_p\pi \);
    a similar argument for \(\psi'\), using 
    \(\psi'^2 = [-\epsilon_pd]\pi_{\EC'}\),
    yields 
    \( [r]\psi' = [p] - \epsilon_p\pi \),
    completing Eq.~\eqref{eq:rpsi}.
    \qed
\end{proof}

\begin{corollary}
    The curves \(\EC\) and \(\EC'\) are ordinary
    if and only if \(r \not= 0\);
    and then
    \[
        |r|
        =
        [ \ZZ[\psi] : \ZZ[\pi_{\EC}] ]
        =
        [ \ZZ[\psi'] : \ZZ[\pi_{\EC'}] ]
        \ .
    \]
\end{corollary}
\begin{proof}
    The curves \(\EC\) and \(\EC'\) are ordinary 
    (not supersingular) 
    if and only if \(p\nmid\trace{\EC}\),
    if and only if \(p\nmid r\)
    (using Eq.~\eqref{eq:dr2-2p-et} and \(p \nmid d\)).
    But the Hasse bound gives \(|\trace{\EC}| \le 2p\),
    so \(|r| \le 2\sqrt{p/d}\); the only \(r\) in this interval 
    divisible by \(p\) is \(0\), proving the first claim.
    If \(\EC\) is ordinary,
    then \(\ZZ[\pi_{\EC}]\) and \(\ZZ[\psi]\)
    are quadratic imaginary orders 
    of discriminant
    \(d^2r^2 - 4dp\) and \(\trace{\EC}^2 - 4p^2 = r^2(d^2r^2 - 4dp)\),
    respectively,
    so \(|r|\) is the conductor of \(\ZZ[\pi_{\EC}']\) in \(\ZZ[\psi]\).
    The same holds for \(\psi'\) and \(\pi_{\EC'}\) on \(\EC'\).
    \qed
\end{proof}

Equation~\eqref{eq:dr2-2p-et}
relates \(r\) to the orders of \(\EC\) and \(\EC'\):
we find
\begin{equation}
    \label{eq:group-order}
        \#\EC(\FF_{p^2})
        =
        (p + \epsilon_p)^2 - \epsilon_pdr^2
        \quad\text{and}\quad
        \#\EC'(\FF_{p^2})
        =
        (p - \epsilon_p)^2 + \epsilon_pdr^2
        \ .
\end{equation}
If \(\EC\) is supersingular, so \(r = 0\),
then~\cite[Theorem~1.1]{Zhu} yields a stronger statement:
\(\EC(\FF_{p^2}) \cong (\ZZ/(p + \epsilon_p)\ZZ)^2\)
and
\(\EC'(\FF_{p^2}) \cong (\ZZ/(p - \epsilon_p)\ZZ)^2\).

\begin{corollary}
    \label{cor:eigenvalue}
    Suppose \(\EC\) is ordinary.
    If \(\G\subseteq\EC(\FF_{p^2})\) 
    (resp.~\(\G'\subseteq\EC'(\FF_{p^2})\))
    is a cyclic subgroup 
    such that \(\psi(\G) \subseteq \G\)
    (resp.~\(\psi'(\G') \subseteq \G'\)),
    then the eigenvalue of \(\psi\) on \(\G\) 
    (resp.~\(\psi'\) on \(\G'\)) is
    \[
        \lambda_\psi
        \equiv
        \frac{p + \epsilon_p}{r}
        \pmod{\#\G}
        \quad
        \text{and}
        \quad
        \lambda_{\psi'}
        \equiv
        \frac{p - \epsilon_p}{r}
        \pmod{\#\G'}
        \ .
    \]
\end{corollary}
\begin{proof}
    Theorem~\ref{th:main} states that \([r]\psi = [p] + \epsilon_p\pi\)
    in \(\End(\EC)\),
    so \(\ker([p] + \epsilon_p\pi)\) contains \(\EC[r]\),
    and hence \([p] + \epsilon_p\pi\) is divisible by \(r\) in
    \(\End(\EC)\): indeed, the quotient is \(\psi\).
    The result follows on restricting to \(\G\);
    the argument for \(\psi'\) and \(\G'\) is the same.
    \qed
\end{proof}

\paragraph{\(\QQ\)-curves: Where from, and why?}

Now we just need a source of quadratic \(\QQ\)-curves of small degree.
Elkies~\cite{Elkies}
shows that all \(\QQ\)-curves
correspond to rational points on certain modular curves:\footnote{%
    The reader unfamiliar with modular curves 
    can get away with the following here:
    the modular curve \(X_0(d)\) parametrizes isomorphism
    classes of (cyclic) \(d\)-isogenies. 
    If \(d\) is prime---which is the only case we need here---then
    there is a unique Atkin--Lehner involution \(\omega\) on \(X_0(d)\):
    its action corresponds to exchanging isogenies \(\phi\) 
    with their duals \(\dualof{\phi}\).
    The quotient \(X_0(d) \to X^*(d) := X_0(d)/\subgrp{\omega}\)
    is a double cover, mapping (the isomorphism class of) an isogeny
    \(\phi\) to 
    the pair
    (of isomorphism classes) 
    \(\{\phi,\dualof{\phi}\}\).
}
Let \(X^*(d)\) be the quotient of
the modular curve \(X_0(d)\)
by all of its Atkin--Lehner involutions.
If \(e\) is a point in \(X^*(d)(\QQ)\)
and \(E\) is a preimage of \(e\) in \(X_0(d)(\QQ(\sqrt{\Delta}))\setminus X_0(d)(\QQ)\)
for some \(\Delta\),
then \(E\) parametrizes (up to \(\QQbar\)-isomorphism) a \(d\)-isogeny
\(\phiK: \ECK \to 
\conj{\ECK}\) over \(\QQ(\sqrt{\Delta})\),
where \(\sigma\) is the involution of \(\QQ(\sqrt{\Delta})/\QQ\).

Luckily enough, for very small \(d\), 
both \(X_0(d)\) and \(X^*(d)\) have genus zero---so not only 
can we
get plenty of rational points on \(X^*(d)\), 
we can get a whole one-parameter family of
\(\QQ\)-curves of degree~\(d\).
Hasegawa
gives explicit universal curves for \(d = 2\), \(3\), and~\(7\)
in~\cite[Theorem 2.2]{Hasegawa}:
for each squarefree integer \(\Delta \not=1\),
every \(\QQ\)-curve of degree \(d = 2, 3, 7\) over \(\QQ(\sqrt{\Delta})\)
is \(\QQbar\)-isomorphic to a rational specialization of one of these
families. 

Crucially, Hasegawa's families for \(d = 2, 3\), and \(7\)
are defined for \emph{any} squarefree
\(\Delta\); so we are free to start by fixing \(p\), before choosing a
\(\Delta\) to suit.
Indeed, 
the particular choice of \(\Delta\) is theoretically irrelevant---since 
all quadratic extensions of \(\FF_{p}\) are isomorphic---so
we may choose any practically convenient
value for~\(\Delta\),
such as one permitting faster arithmetic
in~\(\FF_{p^2} = \FF_{p}(\sqrt{\Delta})\).


Of course,
one might ask why it is necessary to use this characteristic-zero theory
when we could simply search for curves over \(\FF_{p^2}\) 
with a \(d\)-isogeny to their Galois conjugate.
But for low degrees \(d\) where \(X_0(d)\) has genus zero, 
every such curve 
arises as the reduction mod \(p\) of a \(\QQ\)-curve over
\(\QQ(\sqrt{\Delta})\) where \(\Delta\) is a nonsquare mod~\(p\).
Indeed, the curves over \(\FF_{p^2}\) with \(d\)-isogenies to their
Galois conjugates correspond (up to isomorphism) to points 
in \(X^*(d)(\FF_p)\) with preimages in \(X_0(d)(\FF_{p^2})\setminus X_0(d)(\FF_p)\).
But \(X_0(d)\) is isomorphic 
to \(\mathbb{P}^1\) over the ground field in the cases
we consider, 
so every such point lifts trivially to 
a point in \(X^*(d)(\QQ)\) with preimages in
\(X_0(d)(\QQ(\sqrt{\Delta}))\setminus X_0(d)(\QQ)\):
that is, to a \(\QQ\)-curve.

We therefore lose no candidate curves over \(\FF_{p^2}\)
by reducing \(\QQ\)-curves mod \(p\) instead of working directly over
\(\FF_{p^2}\).
What we gain by working with \(\QQ\)-curves 
is some simplicity\footnote{
    Especially since \(\QQ\)-curves (by definition)
    have no non-integer endomorphisms,
    while every elliptic curve over \(\FF_{p^2}\) has complex
    multiplication.
} 
and universality in our proofs,
to say nothing of 
the wealth of mathematical literature to be raided for examples and theorems.  
For example: instead of deriving and proving
defining equations for these families over \(\FF_{p^2}\), 
we can just conveniently borrow Hasegawa's universal \(\QQ\)-curve equations.

\paragraph{GLS as the degenerate case \(d = 1\).}

Suppose
\(\phiK: \ECK \to \conj{\ECK}\)
is an isogeny of degree \(d = 1\):
that is, an isomorphism.
Then \(j(\ECK) = j(\conj{\ECK}) = \conj{j(\ECK)}\),
so \(j(\ECK)\) is in \(\QQ\),
and \(\ECK\) is \(\QQ(\sqrt{\Delta})\)-isomorphic to 
a curve defined over~\(\QQ\).
Suppose, then, that
\(\ECK\) is defined over \(\QQ\)
and base-extended to \(\QQ(\sqrt{\Delta})\):
then \(\ECK = \conj{\ECK}\),
and we can apply our construction 
taking
\(\phiK: \ECK \to \conj{\ECK}\) to be the identity map.
Reducing modulo an inert~\(p\),
we have \(\psi=\pi_p\)
and \(\psi^2 = \pi_p^2 = \pi_{\EC}\),
so \(\psi\) has eigenvalue \(\pm1\) 
on cryptographic subgroups of
\(\EC(\FF_{p^2})\):
clearly, \(\psi\) is of no use to us
for scalar decompositions.
However, 
the twisted endomorphism \(\psi'\)
on~\(\EC'\)
satisfies \((\psi')^2 = -\pi_{\EC'}\),
so
the eigenvalue of \(\psi'\) on cryptographic subgroups
is a square root of \(-1\), which is large enough to yield good scalar
decompositions.
We have recovered the GLS endomorphism
(cf.~\cite[Theorem~2]{GLS}).

While \(\EC'(\FF_{p^2})\) may have prime order,
\(\EC(\FF_{p^2})\) cannot:
the fixed points of \(\pi_p\) form a subgroup of order \(p + 1 - t_0\),
where \(t_0^2 - 2p = \trace{\EC}\)
(and the complementary subgroup
has order \(p + 1 + t_0\)).
Hence,
the largest prime divisor of \(\#\EC(\FF_{p^2})\)
can be no larger than \(\Oh(p)\);
the curve \(\EC'\) can therefore never be twist-secure.

\section{
    Short Scalar Decompositions
}
\label{sec:decompositions}

Before moving on to concrete families and examples,
we will show that the endomorphisms developed in
\S\ref{sec:construction} yield short scalar decompositions.

Suppose \(\G\cong\ZZ/N\ZZ\) is a cyclic subgroup of \(\EC(\FF_{p^2})\) 
such that \(\psi(\G) \subseteq \G\).
Corollary~\ref{cor:eigenvalue} shows that
\(\psi\) acts as 
the eigenvalue
\(
    \lambda_\psi \equiv (p + \epsilon_p)/r \pmod{N} 
\)
(a square root of \(\epsilon_p d\))
on \(\G\).
Given an integer \(m\), we want to 
compute a decomposition
\[
    m = a + b\lambda_\psi \pmod N
\]
so as to efficiently compute
\( [m]P = [a]P \oplus [b]\psi(P) \)
for  \(P\) in \(\G\).
The decomposition 
is not unique: far from it.
The set of all decompositions \((a,b)\) of \(m\)
is the lattice coset \((m,0) + \Lattice\),
where
\[
    \Lattice
    := 
    \subgrp{
        (N,0),(-\lambda_\psi,1)
    }
    \subset \ZZ^2
\]
is the lattice of decompositions of \(0\):
that is, of integer pairs \((a,b)\) 
such that \(a + b\lambda_\psi \equiv 0\pmod{N}\).

We want to find a decomposition of \(m\)  
where \(a\) and \(b\) have minimal bitlength: 
that is, 
where \(\lceil\log_2\|(a,b)\|_\infty\rceil\) 
is as small as possible.
The following algorithm\footnote{
    Algorithm~\ref{alg:decomp}
    differs from the standard technique (cf.~\cite[\S4]{GLV}),
    based on Babai rounding~\cite{Babai},
    in Step~2.
    Instead of choosing \(\vv{c}\) to be the shortest of the four vectors,
    Babai rounding approximates it
    by selecting
    \(\vv{c}' = \roundoff{\alpha}\vv{e}_1 + \roundoff{\beta}\vv{e}_2\)
    (this is the correct choice for most \(m\)).
    In terms of bitlength, this means an excess of one bit in the worst
    case.
}
computes an optimal decomposition of \(m\)
given a reduced basis of~\(\Lattice\).

\begin{algorithm}
    \label{alg:decomp}
    Given a reduced basis \([\vv{b}_1,\vv{b}_2]\) for \(\Lattice\),
    computes a decomposition of minimal bitlength 
    (and bitlength at most \(\lceil{\log_2\|\vv{b}_2\|_\infty}\rceil\))
    for any given integer~\(m\).
    \begin{description}
        \item[Input] An integer \(m\) 
            and a reduced basis \([\vv{b}_1,\vv{b}_2]\) 
            for \(\Lattice = \subgrp{(N,0),(-\lambda_\psi,1)}\). 
        \item[Output] 
            A pair of integers \((a,b)\)
            such that \(m \equiv a + b\lambda_\psi \pmod{N}\)
        \item[Step 1] 
            Let \(\alpha := mb_{22}/N\)
            and \(\beta := -mb_{12}/N\).
        \item[Step 2]
            Let \(\vv{c}\) be the shortest of
            the four vectors
            \( \lfloor{\alpha}\rfloor\vv{b}_1 + \lfloor{\beta}\rfloor\vv{b}_2 \),
            \( \lfloor{\alpha}\rfloor\vv{b}_1 + \lceil{\beta}\rceil\vv{b}_2 \),
            \( \lceil{\alpha}\rceil\vv{b}_1 + \lfloor{\beta}\rfloor\vv{b}_2 \),
            and
            \( \lceil{\alpha}\rceil\vv{b}_1 + \lceil{\beta}\rceil\vv{b}_2 \).
        \item[Step 3]
            Return \((a,b) := (m,0) - \vv{c}\).
    \end{description}
\end{algorithm}
\begin{proof}
    It is easily checked that 
    \((\alpha,\beta)\) is the unique solution in \(\QQ^2\) 
    to the linear system \(\alpha\vv{b}_1 + \beta\vv{b}_2 = (m,0)\) 
    (here we use the fact that \(N = \det\Lattice\)).
    Then \(\vv{c}\) is the closest vector to \((m,0)\) in \(\Lattice\)
    by Theorem~19 of~\cite{Kaib},
    so \(\|(a,b)\|_\infty\) is minimal over all decompositions of~\(m\).
    For the bound on \(\|(a,b)\|_\infty\),
    set \(\vv{c}' := \roundoff{\alpha}\vv{b}_1 + \roundoff{\beta}\vv{b}_2\);
    then \(\|(m,0) - \vv{c}\|_\infty \le \|(m,0) - \vv{c}'\|_\infty\).
    The triangle inequality and \(|x - \roundoff{x}| \le 1/2\)
    for all \(x\) in \(\QQ\)
    imply
    \(
        \|(m,0) - \vv{c}'\|_\infty
        \le
        \max(\|\vv{b}_1\|_\infty,\|\vv{b}_2\|_\infty)
        =
        \|\vv{b}_2\|_\infty
    \).
    \qed
\end{proof}

It remains to precompute a reduced basis for \(\Lattice\).
If \(|\lambda_\psi|\) is not unusually small,
then there exists a reduced basis 
of size \(\Oh(\sqrt{N})\).\footnote{
    General bounds on the constant hidden by the \(\Oh(\cdot)\) 
    appear in~\cite{SCQ},
    but they are far from tight for inseparable endomorphisms.
    Lemma~\ref{lemma:magic-basis}
    gives much better results
    for our endomorphisms
    in cryptographic contexts.
}
Traditionally, we would compute it
using the Gauss reduction or Euclidean algorithms
(cf.~\cite{Kaib}, \cite[\S4]{GLV} and~\cite[\S17.1.1]{Galbraith}),
but in our case
lattice reduction algorithms are unnecessary:
following the approach outlined in~\cite{Smith-basis},
we can immediately write down a reduced basis for 
a large sublattice of \(\Lattice\), 
which coincides with \(\Lattice\) when \(\G = \EC(\FF_{p^2})\).
If \(\G\) has a small cofactor in \(\EC(\FF_{p^2})\), 
then we can easily modify the sublattice basis to give a proper reduced
basis for \(\Lattice\) 
(as we will do in
Lemmas~\ref{lemma:basis-d2}, \ref{lemma:basis-montgomery},
and~\ref{lemma:basis-d3} below).

\begin{lemma}
    \label{lemma:magic-basis}
    The vectors
    \( \vv{e}_1 = \left(p + \epsilon_p, -r\right) \)
    and
    \( \vv{e}_2 = \left(-\epsilon_pdr, p + \epsilon_p\right) \)
    generate a sublattice \(\Lattice_0\subseteq\Lattice\) 
    of index \([\Lattice:\Lattice_0] = \#\EC(\FF_{p^2})/N\).
    In particular, if \(\#\EC(\FF_{p^2}) = N\),
    then \(\Lattice = \Lattice_0\).
    For large \(p\), 
    \begin{itemize}
        \item
            if \(\epsilon_p = -1\),
            then \([\vv{e}_1,\vv{e}_2]\) is reduced;
        \item
            if \(\epsilon_p = 1\),
            then
            \([\vv{e}_1 + \vv{e}_2,\vv{e}_1]\) (if \(r > 0\))
            or \([\vv{e}_1 - \vv{e}_2,\vv{e}_1]\) (if \(r < 0\))
            is reduced.
    \end{itemize}
    In either case, 
    the bitlength of the reduced basis is
    \(
        \lceil{\log_2(p+\epsilon_p)}\rceil
    \).
\end{lemma}
\begin{proof}
    Corollary~\eqref{cor:eigenvalue} 
    implies
    \(r\lambda_\psi \equiv p + \epsilon_p\)
    and \(r\epsilon_pd \equiv (p + \epsilon_p)\lambda_\psi \pmod{N}\),
    so 
    \(\vv{e}_1\)
    and \(\vv{e}_2\) 
    are in \(\Lattice\);
    they are linearly independent,
    so they generate a sublattice.
    The determinant is \((p+\epsilon_p)^2 - \epsilon_pdr^2\),
    which is \(\#\EC(\FF_{p^2})\) by Eq.~\eqref{eq:group-order}.
    Recall that \(r\) is in \(\Oh(\sqrt{p})\) and \(d\) is very small,
    so \(\|\vv{e}_1\|_\infty = \|\vv{e}_2\|_\infty = p + \epsilon_p\).
    The bases satisfy
    Inequality~\eqref{eq:reduced}, and hence are reduced.
    \qed
\end{proof}

\section{
    Endomorphisms from Quadratic \(\QQ\)-curves of Degree 2
}
\label{sec:degree-2}

Let \(\Delta\) be a squarefree integer.
Hasegawa
defines a one-parameter family 
\[
    \label{eq:Hasegawa-2}
    \ECK_{2,\Delta,s}: 
    y^2 
    = 
    x^3 + 2(C_{2,\Delta}(s) - 24)x - 8(C_{2,\Delta}(s) - 16)
\]
of \(\QQ\)-curves
of degree 2
over \(\QQ(\sqrt{\Delta})\)
in~\cite[Theorem 2.2]{Hasegawa}, 
where
\[
    C_{2,\Delta}(s) := 9(1 + s\sqrt{\Delta})
\]
and \(s\) is a free parameter taking values in \(\QQ\).
Observe that
\(
    \conj{\ECK_{2,\Delta,s}} = \ECK_{2,\Delta,-s}
\).

To realise the \(\QQ\)-curve structure,
observe that \(\ECK_{2,\Delta,s}\) has 
a rational 2-torsion point \((4,0)\).
We compute
the normalized quotient isogeny 
\(
    \ECK_{2,\Delta,s}
    \to 
    \ECK_{2,\Delta,s}/\subgrp{(4,0)}
\)
using Eq.~\eqref{eq:Velu-2};
but then
we observe that
\(
    \ECK_{2,\Delta,s}/\subgrp{(4,0)} 
    =
    \twist{\sqrt{-2}}{(\conj{\ECK_{2,\Delta,s}})}
\),
so composing the quotient 
with the twisting isomorphism \(\twistiso{1/\sqrt{-2}}\)
yields a 2-isogeny
\[
    \phiK_{2,\Delta,s} : 
    \ECK_{2,\Delta,s} 
    \longrightarrow 
    \conj{\ECK_{2,\Delta,s}}
\]
defined by the rational map
\[
    \phiK_{2,\Delta,t}:
    (x,y)
    \longmapsto
    \left(
        \frac{-x}{2} - \frac{C_{2,\Delta}(s)}{x-4}
        , 
        \frac{y}{\sqrt{-2}}
        \left(
            \frac{-1}{2} + \frac{C_{2,\Delta}(s)}{(x-4)^2}
        \right)
    \right)
    \ .
\]
(The arbitrary choice of one of the two square roots of
\(-2\) results in an arbitrary sign on \(\phiK_{2,\Delta,s}\).)
Conjugating and composing again,
we find that
\begin{equation}
    \label{eq:epsilon-2}
    \conj{\phiK_{2,\Delta,s}}
    \circ
    {\phiK_{2,\Delta,s}}
    =
    [\epsilon 2]_{\ECK_{2,\Delta,s}}
    \quad
    \text{ where }
    \epsilon 
    = 
    \begin{cases}
        -1 & \text{if }\ \conj{\sqrt{-2}} = \sqrt{-2}  
        \\
        +1 & \text{if }\ \conj{\sqrt{-2}} = -\sqrt{-2} 
    \end{cases}
\end{equation}
---and similarly, 
\(
    {\phiK_{2,\Delta,s}}
    \circ
    \conj{\phiK_{2,\Delta,s}}
    =
    [\epsilon 2]_{\conj{\ECK_{2,\Delta,s}}}
\).

The discriminant of the family \(\ECK_{2,\Delta,s}\) is 
\( 2^9\cdot C_{2,\Delta}(s)^2\cdot\conj{C_{2,\Delta}(s)} \),
and 
\[
    j(\ECK_{2,\Delta,s})
    =
    \frac{ 
        -12^3 (C_{2,\Delta}(s) - 24)^3 
    }{
        C_{2,\Delta}(s)^2\cdot\conj{C_{2,\Delta}(s)} 
    }
\]
(letting \(s \to \infty\),\footnote{%
    ``Letting \(s\to\infty\)''
    has a proper technical meaning here 
    (and also in \S\ref{sec:degree-3}, \S\ref{sec:degree-7}, and
    \S\ref{sec:CM})---even over finite fields.
    The parameter \(s\) is defined by Hasegawa's choice 
    (following Fricke) of a
    rational parametrization of the modular curve \(X_0(2)\):
    that is, a birational map between \(\PP^1\) and \(X_0(2)\).
    Under this parametrization,
    the point at infinity on
    \(\PP^1\) 
    corresponds to the isomorphism class of the \(2\)-isogeny (in fact,
    the endomorphism) \(1 + \iota\),
    where \(\iota\) is an automorphism of order \(4\) of the curve
    with \(j\)-invariant \(1728\);
    this reflects what we find when we put \(s =
    \infty\) in the formula for \(j(\ECK_{2,\Delta,s})\).
} we find
\(j(\ECK_{2,\Delta,\infty}) = 1728\)).
We see that \(\ECK_{2,\Delta,s}\) reduces modulo any inert \(p > 3\)
to give a family of elliptic curves over \(\FF_{p^2}\),
and then every value of \(s\) in \(\FF_p\)
yields an elliptic curve over \(\FF_{p^2}\).
(Proposition~\ref{prop:j-d2} below
shows that at most two of these curves are isomorphic.)

\begin{theorem}
    \label{th:d2}
    Let \(p > 3\) be prime,
    fix a nonsquare
    \(\Delta\) modulo \(p\),
    so 
    \(\FF_{p^2} = \FF_p(\sqrt{\Delta})\),
    and let \(\EC_{2,\Delta,s}\) and \(\phi_{2,\Delta,s}\)
    be the reductions modulo \(p\)
    of \(\ECK_{2,\Delta,s}\)
    and \(\phiK_{2,\Delta,s}\).

    For each \(s\) in \(\FF_p\),
    the curve \(\EC_{2,\Delta,s}/\FF_{p^2}\) 
    has an efficient \(\FF_{p^2}\)-endomorphism
    \[
        \psi_{2,\Delta,s} := \pi_p\circ\phi_{2,\Delta,s}
    \]
    of degree \(2p\)
    such that
    \(
        \psi_{2,\Delta,s}^2 = [\epsilon_p 2]\pi_{\EC_{2,\Delta,s}} 
    \)
    and
    \(
        (\psi_{2,\Delta,s}')^2 = [-\epsilon_p 2]\pi_{\EC_{2,\Delta,s}'}
    \),
    where
    \[
        \epsilon_p 
        := 
        {-\Legendre{-2}{p}}
        =
        \begin{cases}
            -1 & \text{if }\ p \equiv 1, 3 \pmod{8} \ , \\
            +1 & \text{if }\ p \equiv 5, 7 \pmod{8} \ . \\
        \end{cases}
    \]
    There exists an integer 
    \(r\) satisfying
    \(
        2r^2 = 2p + \epsilon_p \trace{\EC_{2,\Delta,s}} 
    \)
    such that
    \[
        [r]\psi_{2,\Delta,s}
        =
        [p] + \epsilon_p\pi_{\EC_{2,\Delta,s}}
        \qquad \text{and} \qquad 
        [r]\psi_{2,\Delta,s}' 
        = 
        [p] - \epsilon_p\pi_{\EC_{2,\Delta,s}'}
        \ ;
    \]
    the characteristic polynomial
    of
    \(\psi_{2,\Delta,s}\)
    and \(\psi_{2,\Delta,s}'\)
    is
    \(
        P_{2,\Delta,s}(T)
        =
        T^2 - 2rT + 2p 
    \).

    In particular,
    if \(\EC_{2,\Delta,s}\) is ordinary and
    \(\G \subseteq \EC_{2,\Delta,s}(\FF_{p^2})\) 
    is a cyclic subgroup of order \(N\)
    such that \(\psi_{2,\Delta,s}(\G) \subseteq \G\),
    then the eigenvalue of \(\psi_{2,\Delta,s}\) on \(\G\) is 
    \[
        \lambda_{2,\Delta,s} 
        \equiv
        \left( p + \epsilon_p \right)/r
        \equiv 
        \pm\sqrt{\epsilon_p 2} 
        \pmod{N}
        \ .
    \]
\end{theorem} 
\begin{proof}
    Apply Theorem~\ref{th:main} and Corollary~\ref{cor:eigenvalue} to 
    \(\phiK_{2,\Delta,s}\)
    using Eq.~\eqref{eq:epsilon-2}.
    \qed
\end{proof}

\begin{proposition}
    \label{prop:j-d2}
    If \(p > 7\), then
    \( \#\big\{ j(\EC_{2,\Delta,s}) : s \in \FF_p \big\} = p \)
    if \(-7\) is a square in~\(\FF_p\),
    and \(p-1\) otherwise.
\end{proposition}
\begin{proof}
    Suppose \(j(\EC_{2,\Delta,s_1}) = j(\EC_{2,\Delta,s_2})\),
    with \(s_1 \not=s_2\).
    Equating the \(j\)-invariants symbolically, we must have
    \(F_0(s_1,s_2) = 2\sqrt{\Delta}F_1(s_1,s_2)\),
    where the polynomials
    \(F_0(T_1,T_2) = (\Delta T_1T_2 + 1)(81 \Delta T_1T_2 - 175) + 49\Delta(T_1 + T_2)^2\)
    and
    \(F_1(T_1,T_2) = (T_1 + T_2)(63\Delta T_1T_2 - 65)\)
    have coefficients in \(\FF_p\).
    If \(s_1\) and \(s_2\) are in \(\FF_p\),
    then \(F_0(s_1,s_2) = F_1(s_1,s_2) = 0\);
    this happens if and only if
    \(s_2 = -s_1\) and 
    either \(s_i\sqrt{\Delta} = \pm1\) 
    (which is impossible) 
    or
    \(\pm\frac{5}{9}\sqrt{-7}\),
    whence the result.
    \qed
\end{proof}

Both
\(\EC_{2,\Delta,s}(\FF_{p^2})\) and \(\EC_{2,\Delta,s}'(\FF_{p^2})\) 
contain points of order 2:
they generate the kernels of \(\psi_{2,\Delta,s}\) 
and \(\psi_{2,\Delta,s}'\).
If 
\( \#\EC_{2,\Delta,s}(\FF_{p^2}) = 2^kN\)
and \( \#\EC_{2,\Delta,s}'(\FF_{p^2}) = 2^{k'}N'\)
with \(N\) and \(N'\) odd,
then 
Eq.~\eqref{eq:twist-cardinalities}
modulo 8
implies that
either \(k = k' = 1\),
or \(k = 2\) and \(k' \ge 3\), or \(k \ge 3\) and \(k' = 2\).
Equation~\eqref{eq:group-order} modulo 3
implies that if \(p \equiv 2 \pmod{3}\) then
either \(\EC_{2,\Delta,s}(\FF_{p^2})\) 
or \(\EC'_{2,\Delta,s}(\FF_{p^2})\) 
contains a point of order 3.
%
%

\paragraph{Optimal decompositions.}
In view of the Pohlig--Hellman--Silver reduction~\cite{Pohlig--Hellman}
and the rational \(2\)-torsion point on \(\EC_{2,\Delta,s}\),
the ``optimal'' situation for discrete-log based cryptosystems
is when \(\EC_{2,\Delta,s}(\FF_{p^2}) \cong \ZZ/2\ZZ\times\G\)
with \(\#\G\) prime
(though for faster arithmetic, we may want a cofactor of 4 instead of 2;
we consider this later).
Lemma~\ref{lemma:basis-d2}
constructs an optimal basis
for the GLV lattice \(\Lattice\) in this case.
We can use this basis in Algorithm~\ref{alg:decomp}
to decompose scalar multiplications in \(\G\)
as \([m]P = [a]P \oplus [b]\psi_{2,\Delta,s}(P)\)
where \(a\) and \(b\) have at most \(\lceil\log_2p\rceil\) bits.

\begin{lemma}
    \label{lemma:basis-d2}
    Suppose 
    \(\EC_{2,\Delta,s}(\FF_{p^2}) \cong \ZZ/2\ZZ\times\ZZ/N\ZZ\)
    with \(N\) odd, 
    and let \(\Lattice = \subgrp{(N,0),(-\lambda_{2,\Delta,s},1)}\).
    Let \(\vv{e}_1\) and \(\vv{e}_2\)
    be defined as in Lemma~\ref{lemma:magic-basis}.
    For large \(p\), 
    \begin{itemize}
        \item
            if \(\epsilon_pr \ge 0\),
            then 
            \([-\vv{e}_2/2,\vv{e}_1+\vv{e}_2/2]\)
            is a reduced basis for \(\Lattice\);
        \item
            if \(\epsilon_pr < 0\),
            then
            \([-\vv{e}_2/2,\vv{e}_1-\vv{e}_2/2]\)
            is a reduced basis for \(\Lattice\).
    \end{itemize}
    In either case,
    the bitlength of the reduced basis is
    \(\lceil{\log_2(p+\epsilon_p-|r|)}\rceil\).
\end{lemma}
\begin{proof}
    Let
    \( \Lattice_0 := \subgrp{\vv{e}_1,\vv{e}_2} \)
    with
    \( \vv{e}_1 := ( p + \epsilon_p, -r ) \)
    and
    \( \vv{e}_2 := ( -\epsilon_p2r, p + \epsilon_p ) \)
    as in Lemma~\ref{lemma:magic-basis};
    then
    \([\Lattice:\Lattice_0] = 2\),
    so exactly one of \(\vv{e}_1/2\), \(\vv{e}_2/2\), 
    or \((\vv{e}_1+\vv{e}_2)/2\) is in~\(\Lattice\).
    Equation~\eqref{eq:group-order} 
    shows that \(2N = (p + \epsilon_p)^2 - 2\epsilon_pr^2\)
    with \(|r| < \sqrt{2p}\);
    since \(N\) is odd, 
    \(r\) must also be odd,
    so \(\vv{e}_2/2\) is in \(\ZZ^2\)
    but \(\vv{e}_1/2\) and 
    \((\vv{e}_1+\vv{e}_2)/2\) are not---and
    hence they cannot be in \(\Lattice\), either.
    We conclude that \(\Lattice =
    \subgrp{\vv{e}_1,\vv{e}_2/2}\).
    Inequality~\eqref{eq:reduced}
    is satisfied
    by
    \(\vv{b}_1 := -\vv{e}_2/2\)
    and \(\vv{b}_2 := \vv{e}_1\pm\vv{e}_2/2\)
    (with the sign chosen according to whether \(\epsilon_pr\) is positive or
    negative),
    so \([\vv{b}_1,\vv{b}_2]\)
    is a reduced basis for \(\Lattice\).
    The longer of the vectors is \(\vv{b}_2\),
    and \(\|\vv{b}_2\|_\infty = p + \epsilon_p-|r|\).
    \qed
\end{proof}

\begin{example}
    \label{ex:d=2}
    Let \(p = 2^{127}-1\) and \(\Delta=-1\).
    Taking \(s = 28106\) 
    in the family \(\EC_{2,-1,s}/\FF_{p}(\sqrt{\Delta})\)
    yields a twist-secure curve at the 128-bit security level:
    we have \(\epsilon_p = 1\) and 
    \(
        \trace{\EC_{2,-1,28106}}
        =
        -272082382382015736940757543628153813996
    \),
    so 
    \begin{align*}
        \#\EC_{2,-1,28106}(\FF_{p^2})
        & 
        = p^2 + 1 - \trace{\EC_{2,-1,28106}}
        = 2\cdot N
        \quad
        \text{and}
        \\
        \#\EC'_{2,-1,28106}(\FF_{p^2})
        & 
        = p^2 + 1 + \trace{\EC_{2,-1,28106}}
        = 2\cdot N'
    \end{align*}
    where
    \(N\) and \(N'\) are 253-bit primes.\footnote{
        We computed the traces for all of our examples
        using a new specialized variant 
        of the SEA algorithm~\cite{Schoof}
        under development with Fran\c{c}ois Morain and Charlotte Scribot,
        implemented in NTL~\cite{NTL}.
    }
    Algorithm~\ref{alg:decomp} and Lemma~\ref{lemma:basis-d2}
    transform 253-bit scalar multiplications
    in \(\EC_{2,-1,28106}(\FF_{p^2})[N]\)
    into 128-bit multiexponentations.
    This value of \(s\) 
    is the ``smallest'' (counting upwards from 1)
    yielding a curve-twist pair such that both curve orders
    are twice a prime.
    The curve coefficients, being linear in~\(s\),
    are relatively small;
    but 
    while small coefficients are important in optimized
    implementations, here 
    this is no more than a happy coincidence---we did not explicitly 
    search for an example 
    with convenient coefficients.  
\end{example}

\paragraph{Montgomery models.}
The curve \(\EC_{2,\Delta,s}\)
has a Montgomery model over \(\FF_{p^2}\)
if and only if \(2C_{2,\Delta}(s)\)
is a square in \(\FF_{p^2}\)
by~\cite[Proposition~1]{Okeya--Kurumatani--Sakurai}---or equivalently, if \(1 + s\sqrt{\Delta}\) is a square in \(\FF_{p^2}\)
(since \(2\) is always a square in \(\FF_{p^2}\)).
Setting
\[
    \BM := (2C_{2,\Delta}(s))^{1/2}
    \qquad
    \text{and}
    \qquad
    \AM := 12/\BM
    \ ,
\]
the birational map
\( 
    (x,y) 
    \mapsto 
    (X/Z,Y/Z) 
    =
    \big(
        (x-4)/\BM ,
        y/\BM^2
    \big) 
\)
takes us from \(\EC_{2,\Delta,s}\) to the projective Montgomery model
\[
    \EC_{2,\Delta,s}^\mathrm{M}
    : 
    \BM Y^2Z 
    = X\big(X^2 + \AM XZ + Z^2\big) 
\]
(we may replace
the term \(B_{2,\Delta}^\mathrm{M}(s)Y^2Z\) in the defining equation
with a conveniently small multiple of \(Y^2Z\),
if desired, by scaling the \(Y\) coordinate).

Montgomery models offer a particularly efficient arithmetic
using only the \(X\) and \(Z\) coordinates~\cite{Montgomery}.
The induced endomorphism on the \((X:Z)\)-line 
is 
\[
    \psi_{2,\Delta,s}^\mathrm{M}
    :
    (X:Z)
    \longmapsto
    \big(
        X^{2p} + \AM^p X^pZ^p + Z^{2p} 
        :
        -2\AM^{p-1}X^pZ^p
    \big)
    \ ;
\]
an implementation of fast scalar multiplication using
\(\psi_{2,\Delta,s}^\mathrm{M}\) is detailed
in~\cite{Costello--Hisil--Smith}.

\paragraph{Twisted Edwards models.}
Every Montgomery model corresponds to a twisted Edwards model,
and vice versa (cf.~\cite{BBJLP} and \cite{Hisil--Wong--Carter--Dawson}).
Indeed, 
\(\EC_{2,\Delta,s}\)
is isomorphic to the twisted Edwards model
\[
    \EC_{2,\Delta,s}^{\textrm{TE}}:
    \big(12 + 2\BM\big) x_1^2 + x_2^2 
    = 
    1 + \big(12 - 2\BM\big) x_1^2x_2^2
\]
via
\( 
    (x,y) 
    \mapsto
    (x_1,x_2)
    =
    \left(
        (x-4)/y
        ,
        (
            x - 4 - \BM 
            )/(
            x - 4 + \BM
        )
    \right)
\).
Composing with \(\psi_{2,\Delta,s}\)
yields an endomorphism 
\(
    \psi_{2,\Delta,s}^\mathrm{TE}
\)
of \(\EC_{2,\Delta,s}^\mathrm{TE}\).

\paragraph{Optimal decompositions for cofactor 4.}
Every curve with a twisted Edwards or Montgomery model 
has order divisible by 4;
indeed,
\(C_{2,\Delta}(s)\) is a square in \(\FF_{p^2}\)
(so \(\EC_{2,\Delta,s}^\mathrm{M}\)
and \(\EC_{2,\Delta,s}^\mathrm{TE}\)
are defined over \(\FF_{p^2}\))
if and only if \(\EC_{2,\Delta,s}\) has full rational
\(2\)-torsion.
The optimal situation for discrete log-based cryptography on
these curves is therefore when
\(\EC_{2,\Delta,s}(\FF_{p^2}) \cong (\ZZ/2\ZZ)^2\times\G\)
with \(\#\G\) prime.
Lemma~\ref{lemma:basis-montgomery} gives a reduced
basis for the GLV lattice in this case,
which we can use in Algorithm~\ref{alg:decomp}
to decompose scalar multiplications in \(\G\)
as \([m]P = [a]P \oplus [b]\psi_{2,\Delta,s}(P)\)
where \(a\) and \(b\) have at most \(\lceil{\log_2p}\rceil - 1\) bits.

\begin{lemma}
    \label{lemma:basis-montgomery}
    Suppose 
    \(\EC_{2,\Delta,s}(\FF_{p^2}) \cong
    (\ZZ/2\ZZ)^2\times\ZZ/N\ZZ\)
    with \(N\) odd,
    and let
    \(\Lattice=\subgrp{(N,0),(-\lambda_{2,\Delta,s},1)}\).
    Let \(\vv{e}_1\) and \(\vv{e}_2\) be defined as in
    Lemma~\ref{lemma:magic-basis}. For large \(p\),
    \begin{itemize}
        \item
            if \(\epsilon_p = 1\)
            and \(r \ge 0\),
            then \([(\vv{e}_1 + \vv{e}_2)/2,\vv{e}_2/2]\)
            is a reduced basis for \(\Lattice\);
        \item
            if \(\epsilon_p = 1\)
            and \(r < 0\),
            then \([(\vv{e}_1 - \vv{e}_2)/2,-\vv{e}_2/2]\)
            is a reduced basis for \(\Lattice\);
        \item
            if \(\epsilon_p = -1\) and \(r \ge 0\),
            then \([\vv{e}_1/2,\vv{e}_2/2]\)
            is a reduced basis for \(\Lattice\);
        \item
            otherwise,
            if \(\epsilon_p = -1\) and \(r < 0\),
            then \([\vv{e}_1/2,-\vv{e}_2/2]\) 
            is a reduced basis for \(\Lattice\).
    \end{itemize}
    In each case,
    the bitlength of the reduced basis is 
    \(\lceil{\log_2(p + \epsilon_p)}\rceil - 1\).
\end{lemma}
\begin{proof}
    The sublattice
    \(
        \Lattice_0 = \subgrp{
            \vv{e}_1,
            \vv{e}_2
        }
    \)
    has index \(4\) in \(\Lattice\).
    Equation~\eqref{eq:group-order} 
    implies
    \((p+\epsilon_p)^2 \equiv 2\epsilon_pr^2 \pmod{4}\),
    and since \(p + \epsilon_p\) is even,
    \(r\) must be even as well;
    so we can replace the relation 
    \(\epsilon_pr\lambda_{2,\Delta,s} \equiv p + \epsilon_p\)
    with
    \(
        \epsilon_p(r/2)\lambda_{2,\Delta,s} \equiv (p + \epsilon_p)/2
        \pmod{N}
    \).
    Hence, both \(\vv{e}_1/2\)
    and \(\vv{e}_2/2\)
    are in \(\Lattice\);
    so they must form a basis for \(\Lattice\).
    The listed bases 
    are combinations of \(\vv{e}_1/2\) and \(\vv{e}_2/2\)
    satisfying Ineq.~\eqref{eq:reduced},
    and are therefore reduced.
    The longest vector in each basis has length
    \((p+\epsilon_p)/2\).
    \qed
\end{proof}

\paragraph{Doche--Icart--Kohel models.}
Doubling-oriented Doche--Icart--Kohel models 
are defined by equations of the form
\(
    y^2 = x(x^2 + D x + 16D) 
\).
These curves have a rational \(2\)-isogeny \(\phi\)
with kernel \(\subgrp{(0,0)}\),
and both
\(\phi\) and its dual \(\dualof{\phi}\) are in a special form
that allows marginally faster doubling
using the factorization \([2] = \dualof{\phi}{\phi}\) 
(see~\cite[\S3.1]{DIK} for details).
Our curves \(\EC_{2,\Delta,s}\) come equipped with a rational \(2\)-isogeny,
so it is natural to try putting them in Doche--Icart--Kohel
form.
The same \(2\)-isogeny plays two r\^oles in this
situation: as a factor of our endomorphism for scalar decomposition,
and as a factor of the doubling map for Doche--Icart--Kohel arithmetic.
We emphasize that these two applications are distinct and complementary, 
and their benefits are cumulative.
We have an isomorphism
from \(\EC_{2,\Delta,s}\) to
the Doche--Icart--Kohel model
\[
    \EC_{2,\Delta,s}^\mathrm{DIK} 
    : 
    v^2 = u\big(u^2 + \tfrac{1152}{C_{2,\Delta}(s)}u 
    + 16\cdot\tfrac{1152}{C_{2,\Delta}(s)}\big)
    \ ,
\]
defined by
\(
    (x,y) 
    \mapsto 
    (u,v) 
    = 
    (
        \alpha(x-4),
        \alpha^{3/2}y
    )
\)
where
\(\alpha = 96/C_{2,\Delta}(s)\);
if \(C_{2,\Delta}(s)\) is not a square in \(\FF_{p^2}\),
then \(\EC_{2,\Delta,s}^\mathrm{DIK}\) 
is 
\(\FF_{p^2}\)-isomorphic to \(\EC_{2,\Delta,s}'\).

\section{
    Endomorphisms from Quadratic \(\QQ\)-curves of Degree 3
}
\label{sec:degree-3}

Let \(\Delta\) be a squarefree integer.
Hasegawa defines a one-parameter family 
\[
    \label{eq:Hasegawa-3}
    \ECK_{3,\Delta,s}: 
    y^2 =
    x^3 - 3(2C_{3,\Delta}(s) + 1)x
    + \big(C_{3,\Delta}(s)^2 + 10C_{3,\Delta}(s) - 2\big)
\]
of \(\QQ\)-curves
of degree 3 over \(\QQ(\sqrt{\Delta})\) in~\cite[Theorem 2.2]{Hasegawa},
where 
\[
    C_{3,\Delta}(s) := 2(1 + s\sqrt{\Delta})
\]
and \(s\) is a free parameter taking values in \(\QQ\).
Observe that
\(
    \conj{\ECK_{3,\Delta,s}} = \ECK_{3,\Delta,-s}
\).

To realize the degree-3 \(\QQ\)-curve structure,
note that
\(x - 3\)
defines an order-3 subgroup
\(\mathcal{S} = \{0, (3,\pm\conj{C_{3,\Delta}(s)})\}\)
of 
\(\ECK_{3,\Delta,s}(\QQ(\sqrt{\Delta}))\).
Computing the normalized quotient isogeny
\(\ECK_{3,\Delta,s} \to \ECK_{3,\Delta,s}/\mathcal{S}\)
using Eqs.~\eqref{eq:Velu-odd-A}, \eqref{eq:Velu-odd-B},
and~\eqref{eq:Velu-odd-map},
we observe that 
\(
    \ECK_{3,\Delta,s}/\mathcal{S}
    = 
    \twist{\sqrt{-3}}{(\conj{\ECK_{3,\Delta,s}})}
\);
so composing the quotient with \(\twistiso{1/\sqrt{-3}}\)
yields an explicit 3-isogeny 
\(\phiK_{3,\Delta,s}: \ECK_{3,\Delta,s} \to \conj{\ECK_{3,\Delta,s}}\)
defined by the rational map
\[
    \phiK_{3,\Delta,s}
    :
    (x,y)
    \longmapsto
    \Big(
        (\phiK_{3,\Delta,s})_x(x),
        \frac{y}{\sqrt{-3}}\frac{d(\phiK_{3,\Delta,s})_x}{dx}(x)
    \Big)
\]
where
\[
    (\phiK_{3,\Delta,s})_x(x)
    =
    -\frac{1}{3}\Big(
        x
        + \frac{12\cdot\conj{C_{3,\Delta}}(s)}{x-3}
        + \frac{4\cdot\conj{C_{3,\Delta}(s)}^2}{(x-3)^2}
    \Big)
    \ .
\]
Conjugating and composing again, we see that
\begin{equation}
    \label{eq:epsilon-3}
    \conj{\phiK_{3,\Delta,s}}\circ\phiK_{3,\Delta,s}
    =
    \epsilon[3]_{\ECK_{3,\Delta,s}}
    \quad
    \text{where}
    \quad 
    \epsilon
    =
    \begin{cases}
        -1 & \text{if }\ \conj{\sqrt{-3}} = \sqrt{-3} 
        \\
        +1 & \text{if }\ \conj{\sqrt{-3}} = -\sqrt{-3}
    \end{cases}
\end{equation}
(and similarly, 
\( 
    \phiK_{3,\Delta,s}\circ\conj{\phiK_{3,\Delta,s}}
    =
    \epsilon[3]_{\conj{\ECK_{3,\Delta,s}}}
\)).


This family has discriminant 
\( 2^4\cdot3^3\cdot C_{3,\Delta}(s)\cdot \conj{C_{3,\Delta}(s)}^3 \)
and \(j\)-invariant
\[
    j(\ECK_{3,\Delta,s})
    =
    \frac{ 
        2^8\cdot3^3\cdot(2C_{3,\Delta}(s) + 1)^3
    }{ 
        C_{3,\Delta}(s)\cdot \conj{C_{3,\Delta}(s)}^3
    }
\]
(letting \(s \to \infty\), 
we see that \(j(\ECK_{3,\Delta,\infty}) = 0\)).
Hence,
\(\ECK_{3,\Delta,s}\) 
reduces modulo any inert \(p > 3\)
to give a family of elliptic curves over \(\FF_{p^2}\),
and then every value of~\(s\) in~\(\FF_p\)
yields an elliptic curve over \(\FF_{p^2}\).
A calculation similar to Proposition~\ref{prop:j-d2}
shows that we get at least \(p-8\) non-isomorphic curves
in this way.

\begin{theorem}
    \label{th:d3}
    Let \(p > 3\) be prime,
    fix a nonsquare 
    \(\Delta\)
    modulo \(p\),
    so 
    \(\FF_{p^2} = \FF_p(\sqrt{\Delta})\),
    and 
    let \(\EC_{3,\Delta,s}\) and \(\phi_{3,\Delta,s}\)
    be the reductions modulo \(p\) of 
    \(\ECK_{3,\Delta,s}\) 
    and \(\phiK_{3,\Delta,s}\). 

    For each \(s\) in \(\FF_p\),
    the curve \(\EC_{3,\Delta,s}/\FF_{p^2}\)
    has an efficient \(\FF_{p^2}\)-endomorphism 
    \[
        \psi_{3,\Delta,s} := \pi_p\circ\phi_{3,\Delta,s}
    \]
    of degree \(3p\), such that
    \(
        \psi_{3,\Delta,s}^2 
        = 
        [\epsilon_p 3]\pi_{\EC_{3,\Delta,s}}
    \)
    and
    \(
        (\psi_{3,\Delta,s}')^2
        = 
        [-\epsilon_p 3]\pi_{\EC_{3,\Delta,s}'}
    \),
    where
    \[
        \epsilon_p 
        := {-\Legendre{-3}{p}}
        =
        \begin{cases}
            +1 & \text{if }\ p \equiv 2 \pmod{3} \ , \\
            -1 & \text{if }\ p \equiv 1 \pmod{3} \ .
        \end{cases}
    \]
    There exists an integer \(r\) satisfying
    \(
        3r^2 = 2p + \epsilon_p\trace{\EC_{3,\Delta,s}}
    \)
    such that
    \[
        [r]\psi_{3,\Delta,s}
        =
        [p] + \epsilon_p\pi_{\EC_{3,\Delta,s}}
        \qquad
        \text{and}
        \qquad
        [r]\psi_{3,\Delta,s}'
        =
        [p] - \epsilon_p\pi_{\EC_{3,\Delta,s}}
        \ ;
    \]
    the characteristic polynomial
    of
    \(\psi_{3,\Delta,s}\) and \(\psi_{3,\Delta,s}'\)
    is
    \(
        P_{3,\Delta,s}(T)
        =
        T^2 - 3rT + 3p
    \).

    In particular,
    if \(\EC_{3,\Delta,s}\) is ordinary 
    and \(\G \subseteq \EC_{3,\Delta,s}(\FF_{p^2})\)
    is a cyclic subgroup of order \(N\)
    such that \(\psi_{3,\Delta,s}(\G) \subseteq \G\),
    then the eigenvalue of \(\psi_{3,\Delta,s}\) on \(\G\) 
    is
    \[
        \lambda_{3,\Delta,s} 
        \equiv
        (p + \epsilon_p) /r
        \equiv
        \pm\sqrt{\epsilon_p 3} \pmod{N}
        \ .
    \]
\end{theorem} 
\begin{proof}
    Follows from 
    Theorem~\ref{th:main}
    and Corollary~\ref{cor:eigenvalue}
    using Eq.~\eqref{eq:epsilon-3}.
    \qed
\end{proof}

\paragraph{Optimal decompositions.}
The kernel of \(\psi_{3,\Delta,s}\)
is generated by the rational points
\((3,\pm\conj{C_{3,\Delta}(s)})\),
so \(\#\EC_{3,\Delta,s}(\FF_{p^2})\) is always divisible by \(3\).
However, 
the nontrivial points in the kernel of the twisted endomorphism \(\psi_{3,\Delta,s}'\) 
are \emph{not} defined over \(\FF_{p^2}\) 
(they are conjugates), 
so it is possible for \(\EC_{3,\Delta,s}'(\FF_{p^2})\) 
to have prime order.

From the point of view of the Pohlig--Hellman--Silver reduction,
the ``most secure'' curves in \(\EC_{3,\Delta,s}/\FF_{p^2}\) 
have
\(\EC_{3,\Delta,s}(\FF_{p^2}) \cong \ZZ/3\ZZ\times\G\),
with \(\G\) of prime order.
Lemma~\ref{lemma:basis-d3} gives an optimal
basis for the GLV lattice \(\Lattice\) in this case
(for a prime-order twist,
the basis of Lemma~\ref{lemma:magic-basis} is already optimal).

\begin{lemma}
    \label{lemma:basis-d3}
    Suppose 
    \(\EC_{3,\Delta,s}(\FF_{p^2}) \cong \ZZ/3\ZZ\times\ZZ/N\ZZ\)
    with \(N\) prime to \(3\), 
    and let \(\Lattice = \subgrp{(N,0),(-\lambda_{3,\Delta,s},1)}\).
    Let \(\vv{e}_1\) and \(\vv{e}_2\) 
    be defined as in Lemma~\ref{lemma:magic-basis}.
    For large \(p\),
    \begin{itemize}
        \item
            if \(\epsilon_pr \ge 0\), 
            then
            \([\vv{e}_2/3,\vv{e}_1+2\vv{e}_2/3]\)
            is a reduced basis of \(\Lattice\);
        \item
            if \(\epsilon_pr < 0\),
            then
            \([\vv{e}_2/3,\vv{e}_1-2\vv{e}_2/3]\)
            is a reduced basis of \(\Lattice\).
    \end{itemize}
    In either case,
    the bitlength of the reduced basis 
    is \(\lceil{\log_2(p+\epsilon_p-2|r|)}\rceil\).
\end{lemma}
\begin{proof}
    The proof is essentially the same as for Lemma~\ref{lemma:basis-d2},
    with \(3\) in place of~\(2\).
    The sublattice \(\subgrp{\vv{e}_1,\vv{e}_2}\)
    has index 3 in \(\Lattice\),
    so exactly one of \(\frac{1}{3}\vv{e}_1\),
    \(\frac{1}{3}\vv{e}_2\),
    \(\frac{1}{3}(\vv{e}_1 + \vv{e}_2)\),
    and
    \(\frac{1}{3}(\vv{e}_1 - \vv{e}_2)\)
    is in \(\Lattice\).
    Equation~\eqref{eq:group-order}
    gives
    \( 3N = (p + \epsilon_p)^2 - 3\epsilon_pr^2 \);
    but \(p \equiv -\epsilon_p \pmod{3}\),
    so \(\frac{1}{3}\vv{e}_2\) is in \(\ZZ^2\).
    On the other hand, \(3\nmid r\) 
    (since otherwise \(3\mid N\)),
    so neither \(\frac{1}{3}\vv{e}_1\)
    nor \(\frac{1}{3}(\vv{e}_1 \pm \vv{e}_2)\)
    is in \(\ZZ^2\).
    Hence \(\subgrp{\vv{e}_1,\frac{1}{3}\vv{e}_2}\)
    is the only lattice in \(\ZZ^2\)
    containing \(\subgrp{\vv{e}_1,\vv{e}_2}\) with index 3,
    so \(\Lattice = \subgrp{\vv{e}_1,\frac{1}{3}\vv{e}_2}\).
    The vectors \(\frac{1}{3}\vv{e}_2\)
    and \(\vv{e}_1 \pm 2\vv{e}_2/3\) 
    satisfy Ineq.~\eqref{eq:reduced},
    so they form a reduced basis for~\(\Lattice\);
    the longest of their components is \(p+\epsilon_p - 2|r|\).
    \qed
\end{proof}

\begin{example}
    \label{ex:d=3}
    Let \(p = 2^{127}-1\);
    then \(\Delta = -1\) is a nonsquare in \(\FF_p\).
    The parameter value \(s = 10400\)
    yields a twist-secure curve at the 128-bit security level:
    \(
        \#\EC_{3,-1,10400}(\FF_{p^2}) = 3\cdot N
    \)
    and
    \(
        \#\EC_{3,-1,10400}'(\FF_{p^2}) = N' 
    \),
    where \(N\) and \(N'\) are 253- and 254-bit primes,
    respectively.
    As in Example~\ref{ex:d=2},
    this is the smallest value of~\(s\) yielding a curve-twist pair 
    with orders in this form.
    Any scalar multiplication 
    in \(\EC_{3,-1,10400}(\FF_{p^2})[N]\) 
    or \(\EC_{3,-1,10400}'(\FF_{p^2})[N']\) 
    can be computed via a 127-bit multiexponentiation
    using 
    Algorithm~\ref{alg:decomp}
    with the basis of Lemma~\ref{lemma:basis-d3}.
    (We warn the reader that 
    here, one of the curve coefficients is quadratic in \(s\);
    so small values of~\(s\) may not yield particularly convenient
    coefficients for serious implementations.)
\end{example}


\paragraph{Doche--Icart--Kohel models.}
We can exploit the \(3\)-isogeny on \(\EC_{3,\Delta,s}\)
for faster tripling (cf.~\cite[\S3.2]{DIK}):
\(\EC_{3,\Delta,s}\)
is isomorphic 
to the tripling-oriented Doche--Icart-Kohel model
\[
    \EC_{3,\Delta,s}^\mathrm{DIK} 
    : 
    v^2 = u^3 + 3\cdot\frac{9}{{C_{3,\Delta}(s)}^p}(u + 1)^2 
\]
via
\(
    (x,y)
    \mapsto
    (u,v)
    =
    \big(
        \alpha(x-3) ,
        \alpha^{3/2}y
    \big)
\)
where \(\alpha := 3C_{3,\Delta}(s)^{-p}\).
This is an \(\FF_{p^2}\)-isomorphism 
if \(C_{3,\Delta}(s)\) is a square in \(\FF_{p^2}\);
otherwise, 
\(\EC_{3,\Delta,s}^{\mathrm{DIK}}\cong_{\FF_{p^2}}\EC_{3,\Delta,s}'\).

\section{
    Endomorphisms from Quadratic \(\QQ\)-curves of Degree 5
}
\label{sec:degree-5}

For \(d = 5\), Hasegawa notes that it is impossible to give a universal
\(\QQ\)-curve for arbitrary squarefree \(\Delta\):
there exists a quadratic \(\QQ\)-curve of degree \(5\) 
over \(\QQ(\sqrt{\Delta})\)
if and only if 
\( \Legendre{5}{p_i} = 1 \)
for every prime  \(p_i \not= 5\) dividing \(\Delta\)
(see \cite[Proposition~2.3]{Hasegawa}).
This restricts our choice of \(\Delta\) for a given \(p\).

The special case \(\Delta = -1\) is particularly interesting:
by the above, there exists 
a family of \(\QQ\)-curves of degree \(5\) over \(\QQ(\sqrt{-1})\),
and every prime \(p \equiv 3\pmod{4}\) is inert in \(\QQ(\sqrt{-1})\).
We work this case out in detail below.
The remaining case \(p \equiv 1 \pmod{4}\) is a straightforward exercise:
given a fixed prime \(p > 5\),
we choose a squarefree \(\Delta\) meeting the condition above,
then apply~\cite[Theorem 2.4]{Hasegawa}
to derive a family of degree-5 \(\QQ\)-curves over
\(\QQ(\sqrt{\Delta})\)
amenable to the construction of~\S\ref{sec:construction}.\footnote{
    In~\cite{Smith-asiacrypt},
    the author suggested that for any \(p \equiv 1 \pmod{4}\)
    one could use \(\Delta=-11\)
    with Hasegawa's parameters in~\cite[Table 6]{Hasegawa}
    in the construction of \S\ref{sec:construction}.
    This is incorrect: by Dirichlet's theorem,
    half of the \(p \equiv 1 \pmod{4}\) 
    are not inert in \(\QQ(\sqrt{-11})\). 
}

Of course, compared with \(d = 2\) and \(3\),
endomorphisms with separable degree \(d = 5\) 
are intrinsically slower.
The chief interest of this family is that 
unlike with \(d = 2\) and \(3\),
here neither the generic curve nor its twist
have rational torsion points,
so it is possible for reductions and their twists 
to both have prime order.

Let \(\ECK_{5,-1,s}\) be the family of elliptic curves over
\(\QQ(\sqrt{-1})\)
defined by
\[
    \ECK_{5,-1,s}
    :
    y^2 = x^3 + A_{5,-1}(s)x + B_{5,-1}(s)
\]
where
\begin{align*}
    A_{5,-1}(s) 
    & := 
    -27s(11s-2)\left( 3(6s^2 + 6s - 1) - 20s(s - 1)\sqrt{-1} \right)
    \ ,
    \\
    B_{5,-1}(s)
    & :=
    54s^2(11s-2)^2\left( (13s^2 + 59s - 9) - 2(s - 1)(20s + 9)\sqrt{-1} \right)
    \ ,
\end{align*}
and \(s\) is a free parameter taking values in \(\QQ\).

The family \(\ECK_{5,-1,s}\)
is a family of \(\QQ\)-curves of degree 5:
the polynomial
\[
    (1 + 2\sqrt{-1})(x - 3s(11s-2)(2-\sqrt{-1}))^2 + 81s(11s-2)(1 + s\sqrt{-1})^2 
\]
defines the kernel
\(\mathcal{S}\)
of a \(5\)-isogeny
\(\phiK_{5,-1,s}: \ECK_{5,-1,s} \to \conj{\ECK_{5,-1,s}}\)
over \(\QQ(\sqrt{-1})\),
which is the composition of
the normalized quotient 
\(\ECK_{5,-1,s} \to \ECK_{5,-1,s}/\mathcal{S}\)
(as in Eqs.~\eqref{eq:Velu-odd-A}, 
\eqref{eq:Velu-odd-B}, and~\eqref{eq:Velu-odd-map})
with the twisting isomorphism
\( \twistiso{5/(1 + 2\sqrt{-1})} \).
Conjugating and composing again, we find 
\begin{equation}
    \label{eq:d5-comp}
    \conj{\phiK_{5,\Delta,s}}\circ\phiK_{5,\Delta,s}
    =
    [5]_{\EC_{5,\Delta,s}} 
    \qquad
    \text{and}
    \qquad
    \phiK_{5,\Delta,s}\circ\conj{\phiK_{5,\Delta,s}}
    =
    [5]_{\conj{\EC_{5,\Delta,s}}}
    \ .
\end{equation}

The family has discriminant
\(
    -2^63^{12}s^3(11s-2)^3(1 + s^2)(1+s\sqrt{-1})^4
\),
and
\[
    j(\ECK_{5,-1,s}) 
    = 
    \frac{
        -64\left(3 (6 s^2 + 6 s - 1) - 20 (s^2 - s) \sqrt{-1}\right)^3
    }{
        (1 + s^2)(1 + s\sqrt{-1})^4
    }
    \ .
\]
Hence, \(\ECK_{5,-1,s}\) is an elliptic curve 
for all \(s\) in \(\QQ\setminus\{0,2/11\}\),
and these \(\ECK_{5,-1,s}\)
have good reduction at any \(p > 5\) inert in \(\QQ(\sqrt{-1})\).
The analogue of Proposition~\ref{prop:j-d2}
shows that we get at least \(p-25\) non-isomorphic curves
in this way.

\begin{theorem}
    \label{th:d5}
    Let \(p\) be a prime congruent to \(3\) modulo \(4\),
    so 
    \(\FF_{p^2} = \FF_{p}(\sqrt{-1})\),
    and let \(\EC_{5,-1,s}\)
    and \(\phi_{5,-1,s}\) 
    be the reductions mod \(p\)
    of \(\ECK_{5,-1,s}\)
    and \(\phiK_{5,-1,s}\). 

    For each \(s \not=0\) or \(2/11\) in \(\FF_p\),
    the curve \(\EC_{5,-1,s}/\FF_{p^2}\)
    has an efficient \(\FF_{p^2}\)-endomorphism 
    \[
        \psi_{5,-1,s} := \pi_p\circ\phi_{5,-1,s}
    \]
    of degree \(5p\) such that
    \(
        \psi_{5,-1,s}^2 = [5]\pi_{\EC_{5,-1,s}}
    \)
    and
    \(
        (\psi_{5,-1,s}')^2 = [-5]\pi_{\EC_{5,-1,s}'}
    \).
    There exists an integer \(r\) satisfying
    \(
        5r^2 = 2p + \trace{\EC_{5,-1,s}}
    \)
    such that
    \[
        [r]\psi_{5,-1,s} 
        = 
        [p] + \pi_{\EC_{5,-1,s}}
        \qquad
        \text{and}
        \qquad
        [r]\psi_{5,-1,s}'
        = 
        [p] - \pi_{\EC_{5,-1,s}}
        \ ;
    \]
    the characteristic polynomial
    of
    \(\psi_{5,-1,s}\) and \(\psi_{5,-1,s}'\)
    is
    \(
        P_{5,-1,s}(T) = T^2 - 5rT + 5p
    \).

    In particular,
    if \(\EC_{5,-1,s}\) is ordinary
    and \(\G \subseteq \EC_{5,-1,s}(\FF_{p^2})\)
    is a cyclic subgroup of order \(N\) such that
    \(\psi_{5,-1,s}(\G) \subseteq \G\),
    then the eigenvalue of \(\psi_{5,-1,s}\) on \(\G\) is
    \[
        \lambda_{5,-1,s} \equiv (p + 1)/r
        \equiv
        \pm\sqrt{5}
        \pmod{N}
        \ .
    \]
\end{theorem}
\begin{proof}
    Follows from Theorem~\ref{th:main}
    and Corollary~\ref{cor:eigenvalue}
    using Eq.~\eqref{eq:d5-comp}.
    \qed
\end{proof}

Reductions of curves in \(\ECK_{5,-1,s}\) may have
prime order, and so can their twists.
In this situation,
Algorithm~\ref{alg:decomp} with the basis of
Lemma~\ref{lemma:magic-basis} 
computes optimal 
scalar decompositions for \(\psi_{5,-1,s}\)
(of bitlength at most \(\lceil{\log_2(p+1)}\rceil\)).

\begin{example}
    Let \(p = 2^{127}-1\) and \(\Delta=-1\).
    Taking \(s = 7930\) in the degree-5 family
    yields a twist-secure curve at the 128-bit security level:
    the trace of \(\EC_{5,-1,7930}\) is
    \[
        \trace{\EC_{2,-1,28106}}
        =
        160084314926568661653252069280514036151
        \ ,
    \]
    so
    \( 
        \#\EC_{5,-1,7930}(\FF_{p^2})
    \)
    and
    \(
        \#\EC'_{5,-1,7930}(\FF_{p^2})
    \)
    are \emph{both} 254-bit primes.
    We transform 254-bit scalar multiplications in
    \(\EC_{5,-1,7930}(\FF_{p^2})\) into a 127-bit multiexponentiations
    using
    Algorithm~\ref{alg:decomp}
    with the basis of Lemma~\ref{lemma:magic-basis}.
    As in Examples~\ref{ex:d=2} and~\ref{ex:d=3},
    this is the smallest value of \(s\) yielding a curve-twist pair
    with both curves of prime order. (Here the curve coefficients
    are quartic and sextic in \(s\), so the smallness of \(s\) has
    little effect on the convenience of the coefficients for
    implementations---however, as we remarked above, this family
    is essentially of theoretical interest.)
\end{example}

\section{
    Endomorphisms from Quadratic \(\QQ\)-curves of Degree 7
}
\label{sec:degree-7}

For completeness, we include a family of \(\QQ\)-curves 
of degree \(7\).  These curves are less interesting for
practical applications, since the higher degree renders
the endomorphism intrinsically slower than the curves
with \(d = 2, 3\), and \(5\).

Let \(\Delta\) be a squarefree integer.
Hasegawa
defines a one-parameter family 
\[
    \ECK_{7,\Delta,s} : 
    y^2 = x^3 + A_{7,\Delta}(s)x + B_{7,\Delta}(s)
\]
of \(\QQ\)-curves of degree 7 over \(\QQ(\sqrt{\Delta})\)
in~\cite[Theorem~2.2]{Hasegawa},
where
\begin{align*}
    A_{7,\Delta}(s) & = 
    -3C_{7,\Delta}(s)(85 + 96s\sqrt{\Delta} + 15s^2\Delta)
    \ ,
    \\
    B_{7,\Delta}(s) & = 
    14C_{7,\Delta}(s)\big(
        9(3s^4\Delta^2 + 130s^2\Delta + 171)
        +
        16(9s^2\Delta + 163)s\sqrt{\Delta}
    \big)
    \ , 
    \\
    C_{7,\Delta}(s) & = 7(27 + s^2\Delta)
    \ ,
\end{align*}
and
\(s\) is a free parameter taking values in \(\QQ\).
Observe that \(\conj{\ECK_{7,\Delta,s}} = \ECK_{7,\Delta,-s}\).

The family \(\ECK_{7,\Delta,s}\) is 
a family of quadratic \(\QQ\)-curves of degree 7.
More explicitly:
\(\ECK_{7,\Delta,s}\) has a subgroup \(\mathcal{S}\) of order \(7\) 
defined by the kernel polynomial
\[
    (x-C_{7,\Delta}(s))^3 
    - 
    4^2(1 - s\sqrt{\Delta})^2C_{7,\Delta}(s)\big[
    3(x-C_{7,\Delta}(s)) + 4(1 - s\sqrt{\Delta})(27 + s\sqrt{\Delta})
    \big]
    \ .
\]
While \(\mathcal{S}\) is defined over \(\QQ(\sqrt{\Delta})\),
none of its nontrivial points are.
Computing the normalized quotient 
\(\ECK_{7,\Delta,s} \to \ECK_{7,\Delta,s}/\mathcal{S}\)
(using Eqs.~\eqref{eq:Velu-odd-A}, \eqref{eq:Velu-odd-B},
and~\eqref{eq:Velu-odd-map})
and composing with the twisting isomorphism \(\twistiso{1/\sqrt{-7}}\) 
yields an explicit \(7\)-isogeny
\(\phiK_{7,\Delta,s} : \ECK_{7,\Delta,s} \to \ECK_{7,\Delta,s}\).
Conjugating and composing again, 
we see that 
\begin{equation}
    \label{eq:epsilon-7}
    \conj{\phiK_{7,\Delta,s}}\circ\phiK_{7,\Delta,s}
    =
    [\epsilon7]_{\EC_{7,\Delta,s}}
    \quad 
    \text{ where }
    \epsilon
    =
    \begin{cases}
        -1 & \text{if }\ \conj{\sqrt{-7}} = \sqrt{-7}  \\
        +1 & \text{if }\ \conj{\sqrt{-7}} = -\sqrt{-7}  \\
    \end{cases}
\end{equation}
(and similarly, 
\(
    \phiK_{7,\Delta,s}\circ\conj{\phiK_{7,\Delta,s}}
    =
    [\epsilon7]_{\conj{\EC_{7,\Delta,s}}}
\)).

The discriminant of \(\ECK_{7,\Delta,s}\)
is 
\(
    2^{12}\cdot3^6\cdot7\cdot
    C_{7,\Delta}(s)^2
    (1 - s^2\Delta)
    (1 - s\sqrt{\Delta})^6
\),
and
\[
    j\big(\ECK_{7,\Delta,s}\big)
    =
    \frac{
        (27+s^2\Delta)\big(85 + 96s\sqrt{\Delta} + 15s^2\Delta\big)^3
    }{
        (1 - s^2\Delta)(1 - s\sqrt{\Delta})^6
    }
\]
(letting \(s \to \infty\), 
we find \(j(\ECK_{7,\Delta,\infty}) = -3375\));
so \(\ECK_{7,\Delta,s}\) reduces modulo any inert \(p > 7\) 
to give a family of elliptic curves
\(\EC_{7,\Delta,s}/\FF_{p^2}\),
and then any value of \(s\) in \(\FF_p\) 
such that \(s^2 \not= -27/\Delta\) yields an elliptic curve over
\(\FF_{p^2}\).
A calculation similar to Proposition~\ref{prop:j-d2}
shows that we get at least \(p-48\) non-isomorphic curves
in this way, when \(p\) is sufficiently large.

\begin{theorem}
    \label{th:d7}
    Let \(p > 7\) be prime,
    fix a nonsquare \(\Delta\)
    modulo \(p\),
    so 
    \(\FF_{p^2} = \FF_p(\sqrt{\Delta})\),
    and let \(\EC_{7,\Delta,s}\)
    and \(\phi_{7,\Delta,s}\)
    be the reductions modulo \(p\)
    of \(\ECK_{7,\Delta,s}\)
    and \(\phiK_{7,\Delta,s}\).

    For each \(s\) in \(\FF_p\) such that \(s^2 \not= -27/\Delta\),
    the curve \(\EC_{7,\Delta,s}/\FF_{p^2}\)
    has an efficient \(\FF_{p^2}\)-endomorphism
    \(\psi_{7,\Delta,s} := \pi_p\circ\phi_{7,\Delta,s}\)
    of degree \(7p\) satisfying
    \[
        \psi_{7,\Delta,s}^2 
        = 
        [\epsilon_p 7]\pi_{\EC_{7,\Delta,s}}
        \quad
        \text{and}
        \quad
        (\psi_{3,\Delta,s}')^2
        = 
        [-\epsilon_p 7]\pi_{\EC_{7,\Delta,s}'}
    \]
    where
    \[
        \epsilon_p := -\Legendre{-7}{p} 
        = 
        \begin{cases}
            +1 & \textrm{if}\ p \equiv 3, 5, 6 \pmod{7}
            \ ,
            \\
            -1 & \textrm{if}\ p \equiv 1, 2, 4 \pmod{7}
            \ .
        \end{cases}
    \]
    There exists an integer \(r\) satisfying
    \(
        7r^2 = 2p + \epsilon_p\trace{\EC_{7,\Delta,s}}
    \)
    such that
    \[
        [r]\psi_{7,\Delta,s}
        =
        [p] + \epsilon_p\pi_{\EC_{7,\Delta,s}}
        \qquad
        \text{and}
        \qquad
        [r]\psi_{7,\Delta,s}'
        =
        [p] - \epsilon_p\pi_{\EC_{7,\Delta,s}}
        \ ;
    \]
    the characteristic polynomial
    of
    \(\psi_{7,\Delta,s}\) and \(\psi_{7,\Delta,s}'\)
    is
    \(
        P_{7,\Delta,s}(T)
        =
        T^2 - 7rT + 7p
    \).

    In particular,
    if \(\EC_{7,\Delta,s}\) is ordinary and
    \(\G \subseteq \EC_{7,\Delta,s}(\FF_{p^2})\) 
    is a cyclic subgroup of order \(N\)
    such that \(\psi_{7,\Delta,s}(\G) \subseteq \G\),
    then the eigenvalue of \(\psi_{7,\Delta,s}\) on \(\G\) is 
    \[
        \lambda_{7,\Delta,s} 
        \equiv
        \left( p + \epsilon_p \right)/r
        \equiv 
        \pm\sqrt{\epsilon_p 7} 
        \pmod{N}
        \ .
    \]
\end{theorem} 
\begin{proof}
    Follows from Theorem~\ref{th:main}
    and Corollary~\ref{cor:eigenvalue}
    using Eq.~\eqref{eq:epsilon-7}.
    \qed
\end{proof}

\section{
    Exceptional CM and 4-dimensional Decompositions
}
\label{sec:CM}

By definition, \(\QQ\)-curves do not have CM.
However, if \(\ECK_s\) is a family of \(\QQ\)-curves 
then some isolated curves in \(\ECK_s\) may have CM.
These exceptional curves are of interest for 4-dimensional
scalar decompositions: they form a natural generalization of the GLV+GLS
curves described by Longa and Sica~\cite{Longa--Sica}.

Briefly: if \(\ECK/\QQ(\sqrt{\Delta})\) 
has 
CM by an order of small discriminant,
then we can 
compute an explicit endomorphism \(\tilde{\rho}\) of \(\ECK\)
of small degree
(using Stark's algorithm~\cite{Stark}, say),
which then yields an efficient endomorphism \(\rho\) 
on the reduction \(\EC\) of \(\ECK\) modulo~\(p\),
exactly as in the GLV construction.
If \(\ECK\) is \(d\)-isogenous to \(\conj{\ECK}\) 
and \(p\) is inert in \(\QQ(\sqrt{\Delta})\), 
then \(\EC\) also has the degree-\(dp\) 
endomorphism \(\psi\) constructed in \S\ref{sec:construction}.  
The endomorphisms \([1], \rho, \psi\), and \(\rho\psi\)
may then be used as a basis for the 4-dimensional decomposition techniques 
elaborated in~\cite{Longa--Sica}.

\paragraph{Practical limitations of 4-dimensional decompositions.}

``\(\QQ\)-curves with CM''
inherit the chief drawback of the GLV construction:
as noted in~\S\ref{sec:intro},
we cannot hope to find secure (and twist-secure) curves 
when \(p\) is fixed.
This scarcity of secure curves is easily explained:
reductions of CM endomorphisms (including GLV endomorphisms)
are \emph{separable}, and
efficient separable endomorphisms have
extremely small degree, so that their (dense) defining polynomials
can be evaluated quickly.\footnote{%
    By dense, we mean that these polynomials have many nonzero terms;
    the cost of their evaluation therefore depends linearly on the degree.
}
But the degree of an endomorphism is 
the norm of the corresponding CM-order element;
and to have non-integers of very small norm, 
the CM-order must have a tiny discriminant.
Up to twists, the number of elliptic curves 
with CM discriminant \(-D\) is 
the class number \(h(-D)\) 
(which is asymptotically in \(\Oh(\sqrt{D})\)).
The six orders containing endomorphisms of degree 
\(\le 3\) have class number 1,
and hence only one corresponding \(j\)-invariant.
For \(-D = -4\), corresponding to \(j = 1728\),
there are two or four \(\FF_{p^2}\)-isomorphism classes;
for \(-D = -3\), corresponding to \(j = 0\),
we have two or six;
and otherwise we have only two.
In particular,
there are at most 18 pairwise non-isomorphic curves over~\(\FF_{p^2}\)
with a nontrivial endomorphism of degree at most 3.

Over a fixed finite field, 
the probability that any of these curves 
will have a secure group order,
let alone be twist-secure, is very low:
roughly speaking, we expect to try \(O(\log^2p)\) random curves
over \(\FF_{p^2}\)
before finding a twist-secure one
(see~\cite{Shparlinski--Sutantyo}, for example, for more accurate heuristics).
In practice, then, we cannot use these curves when \(p\) is fixed for
efficiency.
Higher-dimensional scalar decomposition 
speedups therefore come at the cost of suboptimal field arithmetic:
we pay for shorter loop lengths with comparatively slower field
(and hence group)
operations, to say nothing of a more complicated multiexponentiation
algorithm.

We must therefore choose between 4-dimensional decompositions and faster
underlying field arithmetic.
Here we have chosen the latter,
so we do not treat CM curves in depth.
However, 
we enumerate 
the exceptional CM curves in our families
in Theorem~\ref{th:CM-fibres},
to provide a convenient source of curves for readers interested in
exploring and implementing 4-dimensional techniques.

\paragraph{Exceptional CM curves.}
Any one-dimensional family of
\(\QQ\)-curves has only finitely many exceptional CM curves,
up to isomorphism, 
and it is easy to compute them.

\begin{theorem}
    \label{th:CM-fibres}
    The exceptional CM curves in the families
    \(\ECK_{2,\Delta,s}\),
    \(\ECK_{3,\Delta,s}\),
    \(\ECK_{5,-1,s}\),
    and
    \(\ECK_{7,\Delta,s}\)
    are as follows.
    (In each table,
    if \(s\sqrt{\Delta}\) takes the given value
    then \(\ECK_{d,\Delta,s}\) has CM by the 
    order of discriminant~\(-D_0f^2\),
    where \(-D_0\) is the fundamental discriminant
    and \(f\) is the conductor.)
    \begin{enumerate}
        \item
            The following table
            lists the CM fibres in \(\ECK_{2,\Delta,s}\)
            (completing Quer's list~\cite[\S5]{Quer-1},
            where \(s\sqrt{\Delta} = 0\), \(\pm\frac{5}{9}\sqrt{-7}\),
            and \(\infty\) are missing).
            \begin{center}
                \begin{tabular}{|r|r|@{\ }|r|r|@{\ }|r|r|@{\ }|r|r|}
                    \hline
                    \(s\sqrt{\Delta}\) & \(-D_0f^2\) &
                    \(s\sqrt{\Delta}\) & \(-D_0f^2\) &
                    \(s\sqrt{\Delta}\) & \(-D_0f^2\) &
                    \(s\sqrt{\Delta}\) & \(-D_0f^2\) \\
                    \hline
                    \hline
                    \(\infty\) & \(-4\cdot1^2\) &
                    \(\pm\frac{5}{9}\sqrt{-7}\) & \(-7\cdot1^2\) &
                    \(\pm\frac{1}{2}\sqrt{5}\) & \(-20\cdot1^2\) &
                    \(\pm\frac{5}{18}\sqrt{13}\) & \(-52\cdot1^2\) 
                    \\
                    \cline{3-8}
                    \(\pm\frac{7}{12}\sqrt{3}\) & \(-4\cdot3^2\) &
                    \(0\) & \(-8\cdot1^2\) &
                    \(\pm\frac{2}{3}\sqrt{2}\) & \(-24\cdot1^2\) & 
                    \(\pm\frac{70}{99}\sqrt{2}\) & \(-88\cdot1^2\) 
                    \\
                    \cline{5-8}
                    \(\pm\frac{161}{360}\sqrt{5}\) & \(-4\cdot5^2\) &
                    \(\pm\frac{20}{49}\sqrt{6}\) & \(-8\cdot3^2\) &
                    \(\pm\frac{4}{9}\sqrt{5}\) & \(-40\cdot1^2\) &
                    \(\pm\frac{145}{882}\sqrt{37}\) & \(-148\cdot1^2\) 
                    \\
                    \hline
                    \multicolumn{6}{c|}{ }
                    &
                    \(\pm\frac{1820}{9801}\sqrt{29}\) & \(-232\cdot1^2\) \\
                    \cline{7-8}
                \end{tabular}
            \end{center}
        \item
            The following table
            lists the CM fibres in \(\ECK_{3,\Delta,s}\)
            (completing Quer's list~\cite[\S6]{Quer-2},
            where \(s\sqrt{\Delta} = 0\),
            \(\pm\frac{1}{4}\sqrt{-11}\), 
            \(\pm\frac{5}{2}\sqrt{-2}\),
            and \(\infty\) are missing).
            \begin{center}
                \begin{tabular}{|r|r|@{\ }|r|r|@{\ }|r|r|}
                    \hline
                    \(s\sqrt{\Delta}\) & \(-D_0f^2\) &
                    \(s\sqrt{\Delta}\) & \(-D_0f^2\) &
                    \(s\sqrt{\Delta}\) & \(-D_0f^2\) 
                    \\
                    \hline
                    \hline
                    \(\infty\) & \(-3\cdot1^2\) &
                    \(\pm\frac{5}{2}\sqrt{-2}\)    & \(-8\cdot1^2\) &
                    \(\pm\frac{1}{2}\sqrt{2}\)     & \(-24\cdot1^2\) 
                    \\
                    \cline{3-6}
                    \(0\)                          & \(-3\cdot2^2\) &
                    \(\pm\frac{1}{4}\sqrt{-11}\)   & \(-11\cdot1^2\) &
                    \(\pm\frac{1}{4}\sqrt{17}\)     & \(-51\cdot1^2\) 
                    \\
                    \cline{3-6}
                    \(\pm\frac{5}{9}\sqrt{3}\)     & \(-3\cdot4^2\) &
                    \(\pm\sqrt{5}\)                & \(-15\cdot1^2\) &
                    \(\pm\frac{5}{32}\sqrt{41}\)   & \(-123\cdot1^2\)
                    \\
                    \cline{5-6}
                    \(\pm\frac{9}{20}\sqrt{5}\)    & \(-3\cdot5^2\) &
                    \(\pm\frac{11}{25}\sqrt{5}\)   & \(-15\cdot2^2\) &
                    \(\pm\frac{53}{500}\sqrt{89}\)   & \(-267\cdot1^2\)
                    \\
                    \cline{3-6}
                    \(\pm\frac{55}{252}\sqrt{21}\) & \(-3\cdot7^2\) &
                    \multicolumn{3}{c}{ }
                    \\
                    \cline{1-2}
                \end{tabular}
            \end{center}
        \item
            The only CM fibres 
            in \(\ECK_{5,-1,s}\)
            are
            \(\ECK_{5,-1,1}\)
            (defined over~\(\QQ\))
            and \(\ECK_{5,-1,-9/13}\);
            both have \(j\)-invariant~\(66^3\)  
            and CM by 
            the order of discriminant \(-4\cdot2^2\). 
        \item
            The following table
            lists the CM fibres in \(\ECK_{7,\Delta,s}\).
            \begin{center}
                \begin{tabular}{|r|r|@{\ }|r|r|}
                    \hline
                    \(s\sqrt{\Delta}\) & \(-D_0f^2\) &
                    \(s\sqrt{\Delta}\) & \(-D_0f^2\) 
                    \\
                    \hline
                    \hline
                    \(\infty\) & \(-7\cdot1^2\) &
                    \(\pm\sqrt{5}\) & \(-35\cdot1^2\) 
                    \\
                    \cline{3-4}
                    \(0\) & \(-7\cdot2^2\) &
                    \(\pm\frac{1}{3}\sqrt{13}\) & \(-91\cdot1^2\) 
                    \\
                    \cline{3-4}
                    \(\pm\frac{1}{3}\sqrt{7}\) & \(-7\cdot4^2\) &
                    \(\pm\frac{5}{39}\sqrt{61}\) & \(-427\cdot1^2\)
                    \\
                    \hline
                \end{tabular}
            \end{center}
    \end{enumerate}
\end{theorem}
\begin{proof}
    Suppose \(\ECK/\QQ(\sqrt{\Delta})\) is isogenous to
    \(\conj{\ECK}\). 
    If \(\ECK\) has CM by the order
    of discriminant \(-D_0 f^2\),
    then so does \(\conj{\ECK}\);
    hence, both \(j(\ECK)\) and \(\conj{j(\ECK)} = j(\conj{\ECK})\)
    are roots of the Hilbert class polynomial \(H_{-D_0 f^2}\).
    But \(H_{-D_0 f^2}\) is irreducible over~\(\QQ\),
    so either 
    \(H_{-D_0 f^2}(T) = T - j(\ECK)\)
    with \(j(\ECK) = j(\conj{\ECK})\) in~\(\QQ\),
    or \( H_{-D_0 f^2}(T) = (T - j(\ECK))(T - j(\conj{\ECK}))\).
    Tables~\ref{tab:CM-j-1} and~\ref{tab:CM-j-2} list 
    every quadratic imaginary discriminant \(-D_0 f^2\)
    such that \(\deg H_{-D_0 f^2}\)
    (which is the class number \(h(-D_0 f^2)\))
    is \(1\) or \(2\),
    along with the associated \(j\)-invariants:
    these can be found
    in the Echidna database~\cite{Echidna}, 
    or (re)computed using Magma~(\cite{Magma-Handbook}, \cite{Magma})
    or Sage~\cite{Sage}.
    To find the exceptional CM curves in each of our families,
    we solve for rational \(s\) and squarefree \(\Delta\) such that
    the \(j\)-invariant of the family appears in Table~\ref{tab:CM-j-1}
    or~\ref{tab:CM-j-2}.
    \qed
\end{proof}

Theorem~\ref{th:CM-fibres} gives
a simple alternative construction for some of the curves investigated
by Guillevic and Ionica in~\cite{Guillevic--Ionica}:
the curves \(E_{1,c}\) and \(E_{2,c}\) of~\cite[\S2]{Guillevic--Ionica}
are \(\twist{\sqrt{3}}{\EC_{2,\Delta,s}}\)
with \(c = s\sqrt{\Delta}\)
and \(\EC_{3,\Delta,s}\) with \(c = -2s\sqrt{\Delta}\),
respectively.
The 255-bit curve of~\cite[Ex.~1]{Guillevic--Ionica}
is a twist of \(\EC_{2,5,4/9}\) by \(\sqrt{3}\).
This curve is not twist-secure.

\begin{table}
    \caption{The thirteen quadratic imaginary
        discriminants \(-D_0 f^2\) 
        of class number 1,
        together with the \(j\)-invariants of 
        the elliptic curves over \(\QQbar\)
        with CM by the quadratic order of each discriminant.
    }
    \label{tab:CM-j-1}
    \begin{center}
        \begin{tabular}{|r|l|@{\ }|r|l|@{\ }|r|l|@{\ }|r|l|}
            \hline
            \(-D_0 f^2\) & \(j\)-invariant &
            \(-D_0 f^2\) & \(j\)-invariant &
            \(-D_0 f^2\) & \(j\)-invariant &
            \(-D_0 f^2\) & \(j\)-invariant \\
            \hline
            \hline
            \(-3\cdot1^2\) & \(0\) 
            &
            \(-4\cdot1^2\) & \(12^3\) 
            &
            \(-8\cdot1^2\) & \(20^3\) 
            &
            \(-43\cdot1^2\) & \(-960^3\) 
            \\
            \cline{5-8}
            \(-3\cdot2^2\) & \(2\cdot30^3\) 
            &
            \(-4\cdot2^2\) & \(66^3\) 
            &
            \(-11\cdot1^2\) & \(-2^{15}\) 
            &
            \(-67\cdot1^2\) & \(-5280^3\) 
            \\
            \cline{3-8}
            \(-3\cdot3^2\) & \(-3\cdot20^3\) 
            &
            \(-7\cdot1^2\) & \(-15^3\) 
            &
            \(-19\cdot1^2\) & \(-96^3\) 
            &
            \(-163\cdot1^2\) & \(-640320^3\) 
            \\
            \cline{1-2}
            \cline{5-8}
            \multicolumn{2}{c|@{\ }|}{ }
            &
            \(-7\cdot2^2\) & \(255^3\) 
            \\
            \cline{3-4}
        \end{tabular}
    \end{center}
\end{table}

\begin{table}
    \caption{The twenty-nine quadratic imaginary
        discriminants \(-D_0 f^2\) 
        of class number~2,
        together with the \(j\)-invariants of 
        the elliptic curves over \(\QQbar\)
        with CM by the quadratic order of each discriminant.
    }
    \label{tab:CM-j-2}
    \begin{center}
        \begin{tabular}{|r|l|}
            \hline
            \(-D_0 f^2\) & \(j\)-invariants \\
            \hline
            \hline
            \(-3\cdot4^2\) & \({12\cdot15^3}( 35010 \pm 20213\sqrt{3} ) \) \\
            \(-3\cdot5^2\) & \({-96^3}( 369830 \pm 165393\sqrt{5} ) \) \\
            \(-3\cdot7^2\) & \({-3\cdot480^3}( 52518123 \pm 11460394\sqrt{21} ) \) \\
            \hline
            \(-4\cdot3^2\) & \({3\cdot4^3}( 399849 \pm 230888\sqrt{3} ) \) \\
            \(-4\cdot4^2\) & \({2\cdot3^3}( 761354780 \pm 538359129\sqrt{2} ) \) 
            \\
            \(-4\cdot5^2\) & \({12^3}( 12740595841 \pm 5697769392\sqrt{5} ) \) \\
            \hline
            \(-7\cdot4^2\) & \({15^3}( 40728492440 \pm 15393923181\sqrt{7} ) \) \\
            \hline
            \(-8\cdot2^2\) & \({10^3}( 26125 \pm 18473\sqrt{2} ) \) \\
            \(-8\cdot3^2\) & \({20^3}( 23604673 \pm 9636536\sqrt{6}) \) \\
            \hline
            \(-11\cdot3^2\) & \({-44\cdot16^3}( 104359189 \pm 18166603\sqrt{33} ) \) \\
            \hline
            \(-15\cdot1^2\) & \(-5\cdot{3^3}( 1415 \pm 637\sqrt{5})/2 \) \\
            \(-15\cdot2^2\) & \(5\cdot{3^3}( 274207975 \pm 122629507\sqrt{5} )/2 \) \\
            \hline
            \(-20\cdot1^2\) & \({5\cdot4^3}( 1975 \pm 884\sqrt{5} ) \) \\
            \hline
            \(-24\cdot1^2\) & \({12^3}( 1399 \pm 988\sqrt{2} ) \) \\
            \hline
            \(-35\cdot1^2\) & \({-5\cdot32^3}( 360 \pm 161\sqrt{5} ) \) \\
            \hline
            \(-40\cdot1^2\) & \({5\cdot12^3}( 24635 \pm 11016\sqrt{5} ) \) \\
            \hline
            \(-51\cdot1^2\) & \({-4\cdot48^3}( 6263 \pm 1519\sqrt{17} ) \) \\
            \hline
            \(-52\cdot1^2\) & \({60^3}( 15965 \pm 4428\sqrt{13} ) \) \\
            \hline
            \(-88\cdot1^2\) & \({60^3}( 14571395 \pm 10303524\sqrt{2} ) \) \\
            \hline
            \(-91\cdot1^2\) & \({-96^3}( 5854330 \pm 1623699\sqrt{13} ) \) \\
            \hline
            \(-115\cdot1^2\) & \({-5\cdot96^3}( 48360710 \pm 21627567\sqrt{5} ) \) \\
            \hline
            \(-123\cdot1^2\) & \({-480^3}( 6122264 \pm 956137\sqrt{41} ) \) \\
            \hline
            \(-148\cdot1^2\) & \({60^3}( 91805981021 \pm 15092810460\sqrt{37} ) \) \\
            \hline
            \(-187\cdot1^2\) & \({-68\cdot240^3}( 2417649815 \pm 586366209\sqrt{17} ) \) \\
            \hline
            \(-232\cdot1^2\) & \({60^3}( 1399837865393267 \pm 259943365786104\sqrt{29} ) \) \\
            \hline
            \(-235\cdot1^2\) & \({-5\cdot1056^3}( 69903946375 \pm 31261995198\sqrt{5} ) \) \\
            \hline
            \(-267\cdot1^2\) & \({-4\cdot240^3}( 177979346192125 \pm 18865772964857\sqrt{89} ) \) \\
            \hline
            \(-403\cdot1^2\) & \({-480^3}( 11089461214325319155 \pm 3075663155809161078\sqrt{13} ) \) \\
            \hline
            \(-427\cdot1^2\) & \({-5280^3}( 53028779614147702 \pm 6789639488444631\sqrt{61} ) \) \\
            \hline
        \end{tabular}
    \end{center}
\end{table}

\end{document}